\documentclass[11pt]{article}

\usepackage[utf8]{inputenc}
\usepackage[T1]{fontenc}

\usepackage{fullpage}
\usepackage{mdframed}
\usepackage{euscript}
\usepackage[dvipsnames]{xcolor}

\usepackage[
pagebackref,
colorlinks=true,
urlcolor=blue,
linkcolor=blue,
citecolor=OliveGreen,
]{hyperref}

\usepackage[a4paper]{geometry}

\usepackage{amsmath, amssymb, amsthm,verbatim, graphicx, xcolor,bbm}
\usepackage{braket}
\usepackage{MnSymbol}
\usepackage{tikz}
\usepackage{doi}

\usetikzlibrary{decorations.pathreplacing}
\usepackage{stmaryrd} 
\usepackage{url}

\usepackage{braket}

\makeatletter
\newtheorem*{rep@theorem}{\rep@title}
\newcommand{\newreptheorem}[2]{%
\newenvironment{rep#1}[1]{%
 \def\rep@title{#2 \ref{##1}}%
 \begin{rep@theorem}}%
 {\end{rep@theorem}}}
\makeatother

\newtheorem{theo}{Theorem} 
\newreptheorem{theo}{Theorem}

\newtheorem{lemma}[theo]{Lemma}

\newtheorem{corol}[theo]{Corollary}
\newtheorem{defn}[theo]{Definition}
\newtheorem{rem}[theo]{Remark}
\newtheorem{claim}[theo]{Claim}
\newtheorem{prop}[theo]{Proposition}

\usepackage{thmtools}
\usepackage{thm-restate}

\def\cG{\mathcal{G}}

\def\1{\mathbbm{1}}

\def\eps{\varepsilon}

\DeclareUnicodeCharacter{00A0}{ }

\newcommand{\G}{\EuScript G}

\newcommand{\C}{\EuScript C}
\newcommand{\X}{\EuScript X}
\newcommand{\eQ}{\EuScript Q}
\newcommand{\Q}{\EuScript Q}
\newcommand{\eL}{\EuScript L}
\newcommand{\bC}{\mathbf C}
\newcommand{\bH}{\mathbf H}
\newcommand{\bW}{\mathbf W}
\newcommand{\f}{\mathbb F}

\DeclareMathOperator{\supp}{supp}

\renewcommand{\geq}{\geqslant}
\renewcommand{\leq}{\leqslant}
\def\rk#1{{\rm rk}{ #1}}

\newcommand{\Cay}{{\sf Cay}}
\DeclareMathOperator{\PSL}{PSL}

\begin{document}

\title{Quantum Tanner codes}

\author{Anthony Leverrier\thanks{Inria, France. {\tt anthony.leverrier@inria.fr}} \qquad Gilles Z{\'e}mor\thanks{Institut de Math\'ematiques
de Bordeaux, UMR 5251, France. {\tt zemor@math.u-bordeaux.fr}}}

\date{\today}	

\maketitle

\begin{abstract}
Tanner codes are long error correcting codes obtained from short codes and a graph, with bits on the edges and parity-check constraints from the short codes enforced at the vertices of the graph. Combining good short codes together with a spectral expander graph yields the celebrated expander codes of Sipser and Spielman, which are asymptotically good classical LDPC codes.

In this work we apply this prescription to the left-right Cayley complex that lies at the heart of the recent construction of a $c^3$ locally testable code by Dinur \textit{et al.} Specifically, we view this complex as two graphs that share the same set of edges. By  defining a Tanner code on each of those graphs we obtain two classical codes that together define a quantum code. This construction can be seen as a simplified variant of the Panteleev and Kalachev asymptotically good quantum LDPC code, with improved estimates for its minimum distance. This quantum code is closely related to the Dinur \textit{et al.}\ code in more than one sense: indeed, we prove a theorem that simultaneously gives a linearly growing minimum distance for the quantum code and recovers the local testability of the Dinur \textit{et al.}\ code. 
\end{abstract}


\section{Introduction}
\subsection{Contributions}

Constructing good linear error correcting codes has been a major endeavor of the information theory community since the early existence results of Shannon. In the past decades, codes have found new playgrounds in fields such as program checking, probabilistically checkable proofs or even quantum computing. Such applications require additional properties beside the capability to correct errors.
For instance, a \emph{locally testable code} (LTC) comes with the remarkable feature that it is possible to read a contant number of bits of a noisy message and estimate how far the message is from the code. In a recent breakthrough, Dinur \textit{et al.}\ \cite{DEL21} showed that LTCs exist even when requiring constant rate, constant distance and constant locality. 
For quantum computing, codes must be able to correct two types of errors: (classical) bit flips and (quantum) phase flips. For this reason, one can define \emph{quantum codes} by giving a pair of linear codes $\C_0, \C_1$ that will each correct one type of errors. 
These codes cannot be chosen arbitrarily: they must satisfy some compatibility
condition that reads $\C_1 \supset \C_0^\perp$. Of major interest --- for theory
and experiments --- are quantum low-density parity-check (LDPC) codes
corresponding to the case when both $\C_0 = \ker H_0$ and $\C_1 = \ker H_1$ are defined by \emph{sparse} parity-check matrices $H_0$ and $H_1$. Here, the major open question concerned the existence of good quantum LDPC codes displaying a constant rate and a constant (relative) distance, and was answered positively in the breakthrough work of Panteleev and Kalachev \cite{PK21}, which also provided an alternative construction of classical LTC. 

A specific combinatorial object lies at the heart of the construction of the LTC
of Dinur \textit{et al.}\ \cite{DEL21}: a \emph{square complex}, and more specifically a
\emph{left-right Cayley complex}, which is a generalisation of expander graphs in higher dimension with cells of dimension 0 (vertices), 1 (edges) and 2 (squares, hence the name). We note that left-right Cayley complexes are also equal to balanced products of Cayley graphs, as introduced by Breuckmann and Eberhardt \cite{BE21b}.
Recalling the seminal expander codes of Sipser and Spielman \cite{SS96} building on the ideas of Tanner \cite{T81} to obtain good classical LDPC codes by putting the bits on the edges of an expander graph and the parity-check constraints on its vertices \textit{via} codes of constant size, it is tempting to apply a similar recipe with higher dimensional objects such as square complexes. In that case, the bits are naturally placed on the 2-dimensional cells of the complex (squares) and the constraints are again enforced at the vertices. The novelty then is that the local view of a vertex becomes a matrix of bits, and allows one to put constraints corresponding to a small tensor product code. The consequence is far reaching: the redundancy between the row constraints and column constraints of the small tensor codes translates into a form of robustness, which can then propagate to the entire code through the square complex, giving in particular the LTCs of Dinur \textit{et al.}\ when the square complex is sufficiently expanding. \\

In the present paper, we apply the Tanner code prescription to the quantum case and construct quantum codes with qubits on the squares and constraints on the vertices of a left-right Cayley complex. The condition $\C_1 \supset \C_0^\perp$ prevents the constraints to be those of small tensor codes, but should come from the dual of such tensor codes. We show that if the complex is sufficiently expanding (\textit{e.g.}\ the complex of \cite{DEL21}) and if the small dual tensor codes exhibit robustness (which is true with high probability for random codes \cite{PK21}), then the construction gives rise to a family of asymptotically good LDPC codes displaying constant rate and linear minimal distance.
The present construction borrows a lot from \cite{PK21}, indeed it uses the same
ingredients, since \cite{PK21} also uses Tanner codes and has a left-right
Cayley complex embedded in its construction: the proposed code family may
therefore be seen as a simplified
variant of \cite{PK21}. However, we wish to stress that it also involves a
conceptual shift: the recent series of breakthrough asymptotic constructions of LPDC codes
with large distances \cite{HHO20, PK20, EB21, PK21} arguably relies upon
increasingly refined notions of chain-complex products that improve upon the
simple product idea of \cite{TZ13}. Our approach breaks away from this paradigm
by proposing to take a geometric (square) complex and directly apply to it the
Tanner code strategy. Indeed, we show that the square complex can be viewed as
two ordinary graphs that share the same set of edges, and the two classical
codes $\C_0,\C_1$ that make up the quantum code are simply defined as classical
Tanner codes on these two graphs.
We obtain improved estimates for the minimum distance of
the resulting quantum code, which for any given code rate scales like
$n/\Delta^{3/2+\eps}$, where $n$ is the code length, $\Delta$ is the degree of the underlying Cayley
graphs, and $\eps >0$ can be taken to be arbitrarily close to $0$.

Interestingly, Panteleev and Kalachev showed that there is often a classical LTC hiding behind a good quantum LDPC code, and 
it turns out that we can recover exactly the LTC of Dinur \textit{et al.}\ from our quantum Tanner codes.\footnote{Note however that it is not recovered through the prescription of Lemma 1 in \cite{PK21} which would give another $c^3$-LTC, namely $\ker H_0^T$, with degraded parameters.}
It can indeed be argued that the present construction is the
missing link between the quantum code construction of \cite{PK21} and the LTC construction of
\cite{DEL21}, for we show that the linear minimum distance of our quantum code
and the local testability of the Dinur \textit{et al.}\ code are direct consequences of a
uniting Theorem that is the main technical result of the present work.

\subsection{Context and history}
\paragraph{Locally testable codes.}

A binary linear code is a subspace of $\f_2^n$.
A $\kappa$-locally testable code with $q$ queries 
is a code $C$ that comes with a tester: the tester
requires access to at most $q$ of bits of any given $n$-bit word $x$, accepts all words of the code,
and rejects a word $x \notin C$ with probability $\geq \frac \kappa nd(x,C)$, 
where $d(x,C)$ is the Hamming distance from $x$ to the code, and $\kappa$ is a
constant called the detection probability. 
A $c^3$-LTC, as constructed in \cite{DEL21}, is such that the rate and distance of the code, as well as the number $q$ of queries, are all constant. In that example, the test simply picks a vertex of the left-right Cayley complex and checks whether the local conditions corresponding to the small tensor code are satisfied. 

LTCs were defined in the early 90s \cite{BFL91} and were first studied in the context of probabilistically checkable proofs (PCPs), before being investigated as mathematical objects of interest notably in \cite{GS06}. 
A good overview of the field can be found in \cite{Gol10}. 
Let us note that a third construction of $c^3$-LTC, besides those of \cite{DEL21} and \cite{PK21}, was obtained by Lin and Hsieh \cite{LH22} through a method similar to that of \cite{PK21}, but relying on lossless expanders rather than spectral expanders.

\paragraph{Quantum LDPC codes.}

Defining the distance of a quantum code is slightly more complicated than in the classical case (where it is the minimum Hamming weight of a nonzero codeword). It is the minimum of two distances, $d_X$ and $d_Z$, that basically characterize how well the code behaves against the two possible types of errors occurring in the quantum case: $X$-type errors, also called bit flips (swapping the basis states $|0\rangle$ and $|1\rangle$) and $Z$-type errors, or phase flips (adding a phase $-1$ to the state $|1\rangle$ while acting trivially on $|0\rangle$). Recall that a quantum CSS code \cite{CS96,ste96} is defined by a pair of classical codes $\C_0 = \ker H_0, \C_1 = \ker H_1$ such that $\C_1 \supset \C_0^\perp$ (or equivalently $H_0\cdot H_1^T = 0$).
The distance of the code is then given by $d = \min (d_X, d_Z)$ with
\[ d_X = \min_{w \in \C_0 \setminus \C_1^{\perp}} |w|, \quad d_Z = \min_{w \in \C_1 \setminus \C_0^\perp} |w|.\]
It is worth noting that for a quantum LDPC code, the sparsity of $H_0$ and $H_1$
implies that both $\C_0^\perp$ and $\C_1^\perp$, and therefore $\C_0$ and $\C_1$
contain (many) words of constant weight. In particular, the codes $\C_0$ and $\C_1$ are not asymptotically good, quite the opposite!

The study of quantum LDPC codes arguably started around 1997 with the
paradigmatic surface code construction of Kitaev \cite{kit03} that encodes a
constant number of logical qubits into $n$ physical qubits and achieves a
distance of $\sqrt{n}$. Improving on this scaling turned out to be challenging:
despite an early construction of \cite{FML02} achieving $n^{1/2} \log^{1/4} n$
in 2002, no further progress was made on this question until 2020. In the
meantime, a major development was the idea of taking a special product of
classical codes \cite{TZ13}, which turned out to correspond to the tensor
product of chain complexes that represent the two classical codes, and yielded
quantum codes of contant rate and minimum distance $\Theta(\sqrt{n})$. Things
accelerated quickly in 2020 when the logarithmic dependence of \cite{FML02} was
first improved for constructions based on high-dimensional expanders
\cite{EKZ20, KT20}, and then much more decisively in a series of works
\cite{HHO20, PK20, EB21} introducing various combinations of chain complex products together with graph lifts. Already well known in the classical case, lifts turned out to be crucial to significantly break the $\sqrt{n}$ barrier on the distance of quantum LDPC codes. 
Finally, Panteleev and Kalachev proved the existence of asymptotically good quantum LDPC codes by considering non-abelian lifts of products of (classical) Tanner codes \cite{PK21}.

We also remark that it is possible to define quantum locally testable codes \cite{AE15}. The existence of such codes would have implications in Hamiltonian complexity, which is the quantum version of computational complexity. In particular, such codes would imply the NLTS conjecture formulated by Hastings \cite{has13,EH17}, and which is itself implied by the quantum PCP conjecture \cite{AAV13}.
Current constructions of quantum LTC are still very weak at the moment and far
from sufficient for such applications, however: they only encode a constant
number of logical qubits with a minimum distance bounded by $O(\sqrt{n})$
\cite{has17, LLZ21}.

A very fruitful approach to prove the existence of certain objects is the
probabilistic method and it is indeed very effective to prove that good
classical LDPC codes \cite{Gal62} and good quantum (non LDPC) codes
\cite{CS96,ste96} exist. This method has failed, however, to produce
$c^3$-LTCs or good quantum LDPC codes. On the one hand, it is well known that a
good LDPC code cannot be locally testable since removing a single constraint
will yield another good code and therefore a word violating this single
constraint will actually be far from the code; on the other hand, picking a good
LDPC code for $\C_0$ in the quantum case forces one to choose $\C_1^\perp$ to
contain words of large weight and $\C_1$ will then not be LDPC. For both
problems, it seems essential to enforce some minimal structure, and left-right
Cayley complexes have provided this fitting, long-awaited framework.

The paper is organised as follows. Section~\ref{sec:overview} gives an overview
of the paper, describes the quantum code construction, states the main theorem
and its consequences, with sketches of proofs. It is structured as a stand-alone extended
summary and concludes with some open
problems.
Section~\ref{sec:prelim} is a preliminary to the detailed part of the paper and
introduces the required technical material.
Section~\ref{sec:main} is the core of the paper, giving the detailed
construction of the quantum code, proving the main theorem and showing how it
implies a linear distance for the quantum code and also how it implies that the
Dinur \textit{et al.}\ code is locally testable.
An appendix is devoted to proving the required behaviour of random dual tensor codes.

\section{Overview}\label{sec:overview}

\paragraph{The left-right Cayley complex.}
Let us summarise the construction of the square complex of Dinur \textit{et al.}
\cite{DEL21}. It is an incidence structure $X$ between a set $V$ of vertices,
two sets of edges $E_A$ and $E_B$, that we will refer to as $A$-edges and
$B$-edges, and a set $Q$ of squares (or quadrangles). 
The vertex-set $V$ is defined from a group $G$: it will be useful for us that the
complex is bipartite, \textit{i.e.}\ the vertex set is partitioned as $V=V_0 \cup
V_1$, with $V_0$ and $V_1$ both identified as a copy of the group $G$. Formally,
we set $V_0=G\times\{0\}$ and $V_1=G\times\{1\}$. Next we have two self-inverse
subsets $A=A^{-1}$
and $B=B^{-1}$ of the group $G$: a vertex $v_0=(g,0)\in V_0$ and a vertex
$v_1=(g,1)$ are said to be related by an $A$-edge if $g'=ag$ for some $a\in A$.
Similarly, $v_0$ and $v_1$ are said to be related by a $B$-edge if $g'=gb$ for
some $b\in B$. The sets $E_A$ and $E_B$ make up the set of $A$-edges and
$B$-edges respectively. In other words, the graph $\G_A=(V,E_A)$ is the double cover
of the {\em left} Cayley graph $\Cay(G,A)$ and likewise $\G_B=(V,E_B)$ is the double
cover of the {\em right} Cayley graph  $\Cay(G,B)$.

Next, the set $Q$ of squares is defined as the set of $4$-subsets of vertices of the form
\[
\{(g,0),(ag,1),(gb,1),(agb,0)\}.
\]
A square is therefore made up of two vertices of $V_0,$ two vertices of $V_1$ as
represented on Figure~\ref{fig:square}.

\begin{figure}[h!]
\begin{center}
\begin{tikzpicture}
\draw[thick] (0,0) -- (0,2) -- (2,2) -- (2,0) -- (0,0);
\node[below left] at (0,0) {$g \in V_0$};
\node[above left] at (0,2) {$gb \in V_1$};
\node[above right] at (2,2) {$agb \in V_0$};
\node[below right] at (2,0) {$ag \in V_1$};
\end{tikzpicture}
\end{center}
\caption{Square of the complex, with edges $(g,ag), (agb, gb) \in E_A,
(g,gb), (agb, ag) \in E_B.$ \label{fig:square}
}
\end{figure}
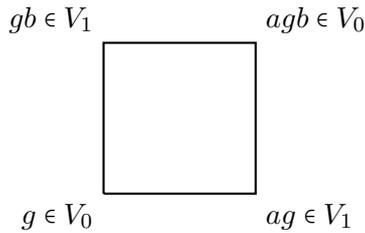

Let us remark that, if we restrict the vertex set to $V_0$, every square is now
incident to only two vertices (those in $V_0$). The set of squares can now be
seen as a set of edges on $V_0$, and it therefore defines a graph that we denote
by $\G_0^{\square}=(V_0,Q)$. Similarly, restricting to vertices of $V_1$ defines the
graph $\G_1^{\square}$, which is an exact replica of $\G_0^{\square}$: both graphs are
defined over a copy of the group $G$, with $g,g,'\in G$ being related by an edge
whenever $g'=agb$ for some $a\in A, b\in B$. We assume
for simplicity that $A$ and $B$ are of the same cardinality $\Delta$.

\paragraph{Tanner codes on the complex $X$.}
Recall the definition of a Tanner code, or expander code, on a graph, 
\cite{T81,SS96}. For $\G=(V,E)$ a regular graph of degree $n_0$ and $C_0$ a binary linear
code of length $n_0$, we define the Tanner code $T(\G,C_0)$ on $\f_2^E$ as the
set of binary vectors indexed by $E$ (functions from $E$ to $\f_2$), such that
on the edge neighbourhood of every vertex $v\in V$, we see a codeword of
$C_0$\footnote{This implies some identification, or map, between the edge
neighbourhood of each vertex and the coordinate set $[n_0]$ on which $C_0$ is
defined}. Sipser and Spielman's celebrated result is that if the graph $\G$ is
chosen from a family of $n_0$-regular expander graphs, and if the base code
$C_0$ has sufficiently large minimum distance and sufficiently large rate, then
the Tanner codes $T(\G,C_0)$ form a family of asymptotically good codes.

 Now for every vertex $v$ of the graph $\G_0^\square=(V_0,Q)$ (or of
$\G_1^\square$) associated to the square complex $X$, there is a natural identification\footnote{Formally, this identification is well-defined provided that the complex satisfies the \emph{Total No-Conjugacy} condition, see Section \ref{subsec:LRCayley} for a precise statement.} of its neighbourhood with the
product set $A\times B$. It therefore makes sense to consider codes $C_0$ on the
coordinate set $A\times B$ that are obtained from two small codes $C_A$ and
$C_B$ of length
$\Delta=|A|=|B|$, defined on coordinate sets $A$ and $B$, respectively.
We will refer to the restriction of an assignment $x \in \f_2^Q$ to the
$Q$-neighbourhood $Q(v)$ of some vertex $v \in V_0$ as the \emph{local view} of $x$ in $v$. The Tanner code construction therefore consists in constraining the local views of $x$ to belong to the code $C_0$.

\paragraph{The locally testable code of Dinur \textit{et al.}}
Let $C_0$ be defined as the tensor code $C_A\otimes C_B$ on the coordinate set
$A\times B$. In other words, this is the code such that for every fixed $b\in B$, we
see a codeword of $C_A$ on $\{(a,b), a\in A\}$, and for every fixed $a$ we see a
codeword of $C_B$ on $\{(a,b), b\in B\}$. The Tanner code $T(\G_0^\square,C_0)$
is exactly the locally testable code of Dinur \textit{et al.}\ \cite{DEL21}.
If $C_A = C_B$ is a linear code with parameters $[\Delta, \rho \Delta, \delta
\Delta]$, the resulting Tanner code $T(\G_0^\square,C_0)$ has length $\Delta^2
|G|/2$, a rate at least $2\rho^2 -1$ and is shown to have a normalized minimum
distance $\geq \delta^2  (\delta - \lambda/\Delta)$, where $\lambda$ is the
(common) second largest eigenvalue of the Cayley graphs $\Cay(G,A)$ and
$\Cay(G,B)$ \cite{DEL21}.

\paragraph{Two Tanner codes that define a quantum LDPC code.}
Besides the base code $C_0=C_A\otimes C_B$ that we have defined over $A\times B$, define the code
$C_1=C_A^\perp\otimes C_B^\perp$. 
Now consider the two Tanner codes $\C_0=T(\G_0^\square,C_0^\perp)$ and
$\C_1=T(\G_1^\square,C_1^\perp)$ that are defined over the same coordinate set $Q$.
We claim that this pair of codes $(\C_0,\C_1)$ satisfies the definition
of a quantum CSS code, namely that $\C_1\supset \C_0^\perp$. Note crucially that we now enforce constraints corresponding to the dual of a tensor code at each vertex.

This last fact is best seen by looking at the generators (in quantum coding
jargon) or parity-checks for these codes. Define a {\em $C_0$-generator} for $\C_0$ (resp.~a $C_1$-generator for
$\C_1$) as a vector of $\f_2^Q$ whose support lies entirely in the
$Q$-neighbourhood $Q(v)$ of a vertex $v$ of
$V_0$ (resp.~$V_1$), and which is equal to a codeword of $C_0$ (resp.~$C_1$) 
on $Q(v)$. (The codes $C_i$ and $\C_i$ should not be confused!) The code $\C_0$ (resp. $\C_1$) is by definition the space of vectors
orthogonal to all $C_0$-generators ($C_1$-generators). Now consider a $C_0$-generator
and a $C_1$-generator on vertices $v_0$ and $v_1$. If the generators have
intersecting supports then the vertices $v_0$ and $v_1$ must be neighbours. If
they are connected by a $B$-edge, then their $Q$-neighbourhoods $A\times B$
share an $A$-set $\{(a,b), a\in A\}$ for a fixed $b$, on which the $C_0$-generator must equal a codeword of $C_A$ and the $C_1$-generator must equal a
codeword of $C_A^\perp$. The two generators must therefore be orthogonal to
each other. We reach the same conclusion analogously if $v_0$ and $v_1$ are
connected by an $A$-edge. 

For reasons of symmetry, we will wish the base codes $C_0$ and $C_1$ to have the
same rate: we will require $C_A$ to have some rate $\rho$ and $C_B$ to have rate
$1-\rho$. In this case, the length of the quantum code $\eQ=(\C_0,\C_1)$ is given by the number of squares in the complex, namely $\Delta^2 |G|/2$, and the number of parity check constraints is $2 \rho(1-\rho) \Delta^2 |G|$. We conclude that the quantum code has rate at least $(2\rho-1)^2$ which is non-zero for every $\rho\neq 1/2$. 

We will show that under the right conditions for the choice of the left-right
Cayley complex $X$ and the component codes $C_A$ and $C_B$, we obtain an
asymptotically good family of quantum codes $\eQ=(\C_0,\C_1)$. The required
conditions on the choice of $X$ are the same as in \cite{DEL21}, namely that the
two Cayley graphs $\Cay(G,A)$ and $\Cay(G,B)$ should be Ramanujan graphs or
almost Ramanujan graphs, \textit{i.e.}\ with a second largest eigenvalue $\lambda\leq c\sqrt{\Delta}$ for some
small constant $c$, and that a non-degeneracy condition (called Total No-Conjugacy or TNC) is satisfied by the
sets $A,B$, ensuring that for all choices of $g \in G, a \in A$ and $b\in B$, it holds that $ag \ne gb$.

It is worth noting that the good asymptotic properties of classical expander
codes follow as soon as the distance of the small code is larger than the second eigenvalue of the expander graph, 
since the normalized minimum distance is known to be at least $\delta (\delta-
\lambda/\Delta) >0$. In the case we are considering however,
the distance of the small code is upper bounded by $\Delta$ while the second eigenvalue scales like $2\Delta$. It is therefore necessary to inspect more closely the structure of the small dual codes $C_A \otimes \f_2^B + \f_2^A \otimes C_B$ in order to get a nontrivial bound on the distance of the quantum Tanner code.

\paragraph{Robustness of the component codes $C_A,C_B$.}
The required conditions for the codes $C_A,C_B$ are that the minimum distances of
$C_A,C_B,C_A^\perp$ and $C_B^\perp$ are all sufficiently large, and that
the two dual tensor codes
$(C_A\otimes C_B)^\perp$ and $(C_A^\perp\otimes C_B^\perp)^\perp$ are sufficiently
{\em robust.} Viewing the sets $A$ and $B$ as the row and column sets
respectively of $\Delta\times\Delta$ matrices, let us say that a dual tensor code $(C_A^\perp\otimes
C_B^\perp)^\perp=C_A\otimes\f_2^B + \f_2^A\otimes C_B$ is
$w$-robust if any codeword $x$ of weight $\leq w$ has its support
included in the union of $|x|/d_A$ columns and $|x|/d_B$ rows, where $d_A$ and
$d_B$ are the minimum distances of $C_A$ and $C_B$. A similar notion
is used both in \cite{DEL21} and \cite{PK21}. In particular, $w$-robustness is equivalent
to a notion of robustness for the tensor code $C_A\otimes C_B$, which loosely
speaking, says that if a vector of $\f_2^{A\times B}$ is close to the
column code $C_A\otimes\f_2^B$ as well as to the row code
$\f_2^A\otimes C_B$, then it must also be close to the tensor code
$C_A\otimes C_B$. We shall be more precise with this notion later on.

To obtain asymptotically good quantum codes, we are however not quite able to
prove the existence of component codes $C_A,C_B$ that yield robust dual tensor
codes for a large enough parameter $w$. To overcome this problem, following
\cite{PK21}, we introduce the following tweak: recall that if $C$ is a code
defined on the coordinate set $S$, and $T \subset S$ is a
subset of $S$, then we may define the punctured code that we will denote
by $(C)_T$, as the set of codewords of $C$ restricted the set of coordinates
$T$. Let us say that the dual tensor code
$C_A\otimes\f_2^B +\f_2^A\otimes C_B$ 
has $w$-robustness {\em with resistance to puncturing} $p$, if, for any subsets $A'
\subset A$ and $B' \subset B$ of cardinality $\geq \Delta - w'$, $w'\leq p$,
it remains 
$w$-robust when punctured outside of the set $A'\times B'$. Equivalently,
if the dual tensor code obtained from the punctured codes $(C_A)_{A'}$ and $(C_B)_{B'}$
is $w$-robust.

Using the method of \cite{PK21}, we will obtain that for any $\eps \in (0,1/2)$
and $\gamma\in (1/2+\eps,1)$,
and for random pairs of codes $C_A,C_B$ of given fixed rates, the associated
dual tensor code $C_A\otimes\f_2^B +\f_2^A\otimes C_B$ is,  with high
probability, $w$-robust with resistance to puncturing $w'$, for
 $w = \Delta^{3/2-\eps}$ and $w'=\Delta^\gamma$.

Our main technical result will now take the following form.

\begin{restatable}{theo}{main}
\label{thm:main} Fix $\eps \in (0,1/2)$, $\gamma\in (1/2+\eps,1)$ and $\delta>0$. 
For any fixed large enough $\Delta$, if the component codes $C_A$ and
$C_B$ have minimum distance $\geq\delta \Delta$ and 
if the dual tensor code $C_A\otimes\f_2^B +\f_2^A\otimes C_B=C_1^\perp$ 
is $w$-robust with $p$-resistance to puncturing for $w = \Delta^{3/2-\eps/2}$ and
$p=\Delta^{\gamma}$, then there exists an infinite family of square complexes $X$ for which the Tanner code
$\C_1=T(\G_1^\square,C_1^\perp)$ 
of length $n=|Q|$ has the following property: 

for any codeword $x\in\C_1$ of non-zero weight $< \delta n/4\Delta^{3/2+\eps}$, there exists a vertex
$v\in V_0$, and a codeword $y$ of $\C_1$ entirely supported by the $Q$-neighbourhood
of $v$, on which it coincides with a codeword of the tensor code
$C_A\otimes C_B$, and such that $|x+y|<|x|$.
\end{restatable}

We recall from \cite{DEL21} that
there exists an infinite sequence of degrees $\Delta$ (namely $q+1$, for $q$ an
odd prime power), such that for each fixed degree $\Delta$, there exists an infinite
family of left-right Cayley complexes (over the groups $G=\PSL_2(q^i)$),
satisfying the TNC condition, for which both the (left) Cayley graph $\Cay(G,A)$
and the (right) Cayley graph $\Cay(G,B)$ are Ramanujan graphs. These complexes
provide the infinite families of Theorem~\ref{thm:main}.
Let us also mention that randomly chosen codes $C_A$ and $C_B$ will typically achieve the requirements of the above theorem. In the quantum case, the codeword $y$ will in fact belong to $\C_0^\perp$.

\paragraph{Sketch of proof of Theorem~\ref{thm:main}.}

Let $x \in \C_1$ be a codeword of sufficiently low weight. It induces a subgraph $\G_{1,x}^\square$ of $\G_1^\square$
 with vertex set $S \subset V_1$. The local view for any $v \in S$ corresponds
to a codeword of $C_1^\perp = C_A\otimes\f_2^B +\f_2^A\otimes C_B$. The
$w$-robustness of this code guarantees that codewords of weight less than $w =
\Delta^{3/2-\eps}$ have a support restricted to a small number of rows and
columns: in particular, and this is the crucial consequence of the robustness
property, it implies that the view restricted to any column (or row) is at
distance $O(\Delta^{1/2-\eps})$ from a word of $C_A$ (or $C_B$).

Let us call \emph{normal vertices} the vertices of $S$ with degree less than $\Delta^{3/2-\eps}$ in $\G_{1,x}^\square$. 
Expansion in $\G_1$ ensures that the set of exceptional (\textit{i.e.}\ not normal) vertices is small compared to $S$ and we will neglect it in this sketch. Dealing with exceptional vertices, however, is slightly technical since their number is not \emph{that} small, and this is the reason why we will require robust codes that are resistant to puncturing (in order to discard rows or columns belonging to an exceptional vertex).

Define $T \subset V_0$ to be vertices of $V_0$ sharing a column (or a row) with a normal vertex in $S$, 
and such that the local view on this column (or row) is close to a nonzero codeword of $C_A$ (or $C_B$).
 The global codeword $x$ induces a subgraph of $\G_A\cup\G_B$ in which the
vertices of $T$ have a degree $\Omega(\Delta)$.
Expansion in the graph $\G_A \cup \G_B$ then implies that when $|x|$ is too
small, there can't be too many vertices in $T$, which in turn implies that a typical vertex in $T$ must be adjacent to a large number 
$\Omega(\Delta)$ of vertices in $S$. In other words, the rows and columns that
are close to codewords of $C_A$ and $C_B$ in the local views of normal vertices
of $S$, must cluster around the vertices of $T$. So the local view of a typical
vertex of $T$
consist of many columns (or rows) containing almost undisturbed codewords of
$C_A$ (or $C_B$). 
The robustness of $C_A \otimes C_B$ then implies that the local view of such a
vertex must be $\Delta^{3/2-\eps}$-close to a codeword $y$ of the tensor code, which
cannot be the zero codeword since the local view has weight $\Omega(\Delta^2)$.
Adding this codeword will decrease the weight of~$x$. 

\paragraph{Asymptotically good quantum codes.}
A straightforward consequence of Theorem~\ref{thm:main} is that the quantum code
$\eQ=(\C_0,\C_1)$ described above has constant rate and minimum distance linear
in its length $n$. 
\begin{theo}\label{thm:good-qLDPC}
Let $X$ be the infinite family of square complexes from Theorem \ref{thm:main}.
For any $\rho \in (0,1/2)$, $\eps\in(0,1/2)$ and $\delta>0$ satisfying
 $-\delta \log_2 \delta - (1-\delta)\log_2(1-\delta) < \rho,$
randomly choosing $C_A$ and $C_B$ of rates
$\rho$ and $1-\rho$ yields, with probability $>0$ for $\Delta$ large enough,
an infinite sequence of quantum codes $\eQ=(\C_0,\C_1)$ of rate $(2\rho-1)^2$,
length $n$ 
and minimum distance $\geq \delta n/4 \Delta^{3/2+\eps}$.
\end{theo}

\noindent
For a more precise statement see Theorems~\ref{thm:123} and \ref{thm:goodcodes}.

\begin{proof}[Sketch of proof of Theorem~\ref{thm:good-qLDPC}]
Recall that the minimum distance of a quantum code $(\C_0,\C_1)$ is the smallest
weight of a vector that is either in $\C_0\setminus \C_1^\perp$ or in
$\C_1\setminus C_0^\perp$. The vector $y$ in Theorem~\ref{thm:main} is in
$\C_0^\perp$, so by applying repeatedly the existence a such a $y$, we obtain
that any codeword $x$ of $\C_1$ of weight $<\delta n/4\Delta^{3/2+\eps}$ must belong to
$\C_0^\perp$. To similarly bound from below the weight of a vector in $\C_0$ but
not in $\C_1^\perp$, one must apply Theorem~\ref{thm:main} to the code $\C_0$
instead of $\C_1$, which just means that we need to ensure that the distance and
robustness properties
required of $C_A$ and $C_B$ are also satisfied by $C_A^\perp$ and $C_B^\perp$.
We choose $C_A$ and $C_B$ randomly by picking a uniform random parity-check matrix for
one code and a uniform random generator matrix for the other, so that properties
typically satisfied by the pair $C_A,C_B$ will also be satisfied by
the pair $C_A^\perp,C_B^\perp$.
\end{proof}

\paragraph{Recovering the local testability of the construction of~\cite{DEL21}.}

Recall that the tester picks randomly a vertex $v\in V_0$ and checks whether the
local view $x_v$ of the input vector $x\in\f_2^Q$ belongs to the small tensor
code $C_0 = C_A \otimes C_B$. Proving local testability amounts to proving that
the distance $d(x,\C)$ of $x$ to the LTC $\C$ is always proportional to the
size $|S|$ of the subset $S\subset V_0$ of vertices for which $x_v$ is not a
codeword of $C_0$. To this end we consider the collection $\mathcal{Z}=(c_v)_{v\in V_0}$
of all the closest $C_0$-codewords to the local views $x_v$ of $x$. Conflating
the small vector $c_v\in C_0$ with a vector in $\f_2^Q$ that equals $c_v$ on the
$Q$-neighbourhood of $v$ and is zero elsewhere, we then
define the vector $z=\sum_{v\in V_0}c_v\in\f_2^Q$ that we call the mismatch of the
collection $\mathcal{Z}$. If we have $z=0$ then $\mathcal{Z}$ is the
collection of local views of a global codeword $c\in \C$, and its distance to
$x$ must be proportional (up to a $\Delta^2$ factor) to $|S|$. Otherwise the key
observation is that the local views of the mismatch $z$ at all the vertices of
$V_1$, must consist of codewords of $C_A\otimes\f_2^B+\f_2^A\otimes C_B$, in
other words $z$ must belong to the code $\C_1$ of Theorem~\ref{thm:main}.
The same Theorem~\ref{thm:main} then states that there exists a vector $y_v\in\f_2^Q$,
that is equal to a codeword of $C_0$ on the $Q$-neighbourhood of some vertex
$v\in V_0$, and such that $|z+y_v|<|z|$. Consequently, if one replaces $c_v$ by
$c_v+y_v$ in the collection $\mathcal{Z}$, one reduces the weight of the
mismatch, and by repeatedly applying the procedure (which we can think of as
decoding the mismatch), we obtain a updated list $\mathcal{Z}$ than coincides
with the set of local views of a codeword of $\C$ whose distance to $x$ is again
easily shown to be proportional to the size of the original set $S$.

\paragraph{Comments and open problems.}

A natural follow up problem is to devise an efficient decoding algorithm for
quantum Tanner codes. At the moment, to the best of our knowledge, the only
efficient decoder for quantum codes that corrects adversarial errors of weight
larger than $\sqrt{n}$ is that of \cite{EKZ20} which can correct errors of
weight $\Omega(\sqrt{n} \log n)$. 

While the construction of quantum Tanner codes described above is arguably conceptually simpler than the balanced product and lifted product constructions of \cite{BE21b, PK21}, it comes at the price of larger weights for the generators, namely $\Theta(\Delta^2)$ instead of $\Theta(\Delta)$. Even if they are constant, it would be very useful to decrease these weights as much as possible and it would be interesting to explore how the weight reduction technique of Hastings \cite{has21} can help on this issue. 

It would be naturally very interesting to find other complexes, besides
left-right Caley complexes, on which the Tanner construction can be applied to
yield good families of quantum LDPC codes.

It  would also be desirable to have completely explicit constructions.
Reed-Solomon codes (binarised versions) come close, but tensor codes of
Reed-Solomon codes are only known to have the required robustness when the sum
of the rates of the component codes is $<1$ \cite{PS94}, which is insufficient
for the above constructions.
In a similar vein, we remark that if one could improve the robustness of the
component dual tensor codes to values above $\Delta^{3/2}$, we could improve the
dependence on $\Delta$ of the relative minimum distance of the quantum code, potentially up
to $\Omega(1/\Delta)$. 
We also remark that we cannot hope to improve the dependence on
$\Delta$ of the relative minimum distance above $O(1/\Delta)$. Indeed, 
we will  prove the existence of words of
$\C_1\setminus \C_0^\perp$ of weight less than $n/\Delta$ (see Section~\ref{sec:goodLDPC}
 for details). This shows that there is little room for improvement in our
estimation of the quantum code minimum distance.

\paragraph{Acknowledgements.}
We would like to thank Benjamin Audoux, Alain Couvreur, Shai Evra, Omar Fawzi,
Tali Kaufman, Jean-Pierre Tillich, and Christophe Vuillot for many fruitful discussions on quantum codes over the years. We also thank Max Hopkins for spotting a technical error in the application of Proposition~\ref{prop:tensor} in an earlier version of this manuscript.
We acknowledge support from the Plan France 2030 through the project ANR-22-PETQ-0006. AL acknowledges support from the ANR through the QuantERA project QCDA, and GZ acknowledges support from the ANR through the project QUDATA, ANR-18-CE47-0010.


\section{Preliminaries}\label{sec:prelim}
\subsection{Expander Graphs}
Let $\G=(V,E)$ be a graph. Graphs will be undirected but may have multiple
edges. 
For $S,T \subset V$, let $E(S,T)$ denote the multiset of edges with one endpoint in
$S$ and one endpoint in $T$\footnote{with the convention that an edge with both
endpoints in $S\cap T$ appears twice in $E(S,T)$}.
Let $\G$ be a $\Delta$-regular graph on $n$ vertices, and let
$\Delta=\lambda_1\geq\lambda_2\geq \ldots \geq \lambda_n$ be the eigenvalues of the
adjacency matrix of $\G$. For $n\geq 3$, we define $\lambda(\G):= \max\{|\lambda_i|,
\lambda_i\neq \pm \Delta\}$. 
The connected graph $\G$ is said to be Ramanujan if $\lambda(\G)\leq 2\sqrt{\Delta-1}$.

We recall the following version of the expander mixing lemma (see e.g. \cite{HLW06}).
\begin{lemma}[Expander mixing lemma] \label{lem:mixing}
For a $\Delta$-regular non-bipartite, connected graph $\cG$ and any sets $S, T \subset V(\cG)$, it holds that
\[ |E(S,T) | \leq \frac{\Delta}{|V|} |S| |T| + \lambda(\G) \sqrt{|S| |T|}.\]
For a $\Delta$-regular bipartite connected graph $\cG$ over the vertex set $V=V_0\cup V_1$
and any sets $S\subset V_0$, $T\subset V_1$, it holds that 
\[ |E(S,T) | \leq \frac{\Delta}{|V_0|} |S| |T| + \lambda(\G) \sqrt{|S| |T|}.\]

\end{lemma}

\subsection{Left-right Cayley complexes.}
\label{subsec:LRCayley}

A {\em left-right Cayley complex} $X$ is introduced in \cite{DEL21} from a group
$G$ and two sets of generators $A=A^{-1}$ and $B=B^{-1}$. As in \cite{DEL21} we
will restrict ourselves, for the sake of simplicity, to the case
$|A|=|B|=\Delta$. The complex is made up of vertices, $A$-edges, $B$-edges, and
squares. The vertex set is $G$,
the $A$-edges are pairs of vertices of the form $\{g,ag\}$ and $B$-edges are of
the form $\{g,gb\}$ for $g\in
G,a\in A,b\in B$. A {\em square} is a set of group elements of the form $\{g,ag,gb,agb\}$.
The {\em Total No-Conjugacy} condition (TNC) requires that
\begin{align}\label{eqn:TNC}
\forall a \in A, b \in B, g \in G, \quad ag \ne gb.
\end{align}
This condition ensures that a square contains exactly $4$ distinct vertices and
that every vertex is incident to exactly $\Delta^2$ squares. For a vertex $v$, the set of incident squares is called the \emph{link} of $v$, and denoted $X_v$. The TNC condition implies that the link of any vertex is in bijection with the set $A \times B$ (see Claim 3.7 of \cite{DEL21}). 
We will naturally refer to sets of the form $\{a\} \times B$ as rows and sets $A \times \{b\}$ as columns. 

We recall from \cite{DEL21} that there exists an infinite sequence of degrees
$\Delta$ (namely $q+1$, for $q$ an odd prime power) such that each fixed degree $\Delta$, there exists an infinite
family of left-right Cayley complexes (over the groups $G=\PSL_2(q^i)$), satisfying the TNC condition, for which
both the (left) Cayley graph $\Cay(G,A)$ and the (right) Cayley graph
$\Cay(G,B)$ are Ramanujan graphs.

When the Cayley graphs $\Cay(G,A)$ and 
$\Cay(G,B)$ are not bipartite (as is the case of the above
family\footnote{Thanks to Shai Evra for spelling this out.}), it will be convenient for us to make them so by
replacing them by their double covers. So we make two copies $V_0=G\times\{0\}$ and 
$V_1=G\times\{1\}$ of $G$ and define the graphs $\G_A=(V=V_0\cup V_1,E_A)$ and
$\G_B=(V,E_B)$
with the edge set $E_A$ made up of pairs $\{(g,0),(ag,1)\}$, $a\in A$,  and $E_B$ consisting
of the pairs $\{(g,0),(gb,1)\}$, $b\in B$.

Finally, the set $Q$ of squares of the square complex $X$ is defined as the set of $4$-subsets of vertices of the form
\[
\{(g,0),(ag,1),(gb,1),(agb,0)\}.
\]

Let us introduce two further graphs that exist on the complex $X$. The first is
just the union of $\G_A$ and $\G_B$, and we denote it by $\G^\cup = (V,E_A\cup
E_B)$. The second graph we denote by $\G^\square=(V,E^\square)$: it puts an edge
between all pairs of vertices of the form $\{(g,i),(agb,i)\}$, $g\in G,a\in A,b\in
B, i=0,1$. The graph $\G^\square$ is therefore made up of two connected
components, on $V_0$ and $V_1$, that we denote by $\G_0^\square$ and
$\G_1^\square$. We note that $\G^\square$ is regular of degree $\Delta^2$, and
may have multiple edges. 

If $\G_A$ and $\G_B$ are Ramanujan, then $\G^\cup$ and $\G^\square$ inherit most
of their expansion properties. Specifically:

\begin{lemma}\label{lem:lambda}
Assume that $\G_A$ and $\G_B$ are Ramanujan graphs, then 
\[ \lambda(\G^\cup) \leq 4\sqrt{\Delta}, \quad
\lambda(\G_0^\square) \leq 4\Delta, \quad \lambda(\G_1^\square)\leq 4\Delta.\]
\end{lemma}

\begin{proof}
Let $M_A$ and $M_B$ be the adjacency matrices of $\G_A$ and $\G_B$. Since these graphs are bipartite, they admit two eigenvalues $\Delta$ and $-\Delta$ and the remaining eigenvalues have an absolute value less than $2\sqrt{\Delta}$. 
The adjacency matrices of $\G^\cup$ and $\G^\square$ are respectively $M_A + M_B$ and $M_A M_B = M_B M_A$.
Since $M_A$ and $M_B$ commute, they can be made to have the same eigenspaces:
therefore the eigenvalues of $M_A+M_B$ are the sum of the eigenvalues of $M_A$
and $M_B$ and the eigenvalues of $M_AM_B$ are the products of the eigenvalues of
$M_A$ and $M_B$. The eigenvalue $-\Delta$ has the same eigenspace in $M_A$ and
$M_B$, therefore the eigenvalue $\Delta^2$ has an eigenspace of dimension 2 in
$M_AM_B$, meaning that $\G^\square$ splits into the two connected components
$\G_0^\square$ and $\G_1^\square$. 
\end{proof}

\paragraph{The quadripartite version.} A way of avoiding the somewhat cumbersome
TNC condition is to make the complex quadripartite rather than simply bipartite.
In this case we construct the vertex set $V$ as the disjoint union of four
copies of the group $G$: we set $V=V_0\cup V_1$ with $V_0=V_{00}\cup V_{11}$ and
$V_1=V_{10}\cup V_{01}$, where $V_{ij}=G\times\{i,j\}$, $i,j\in\{0,1\}$. The set
$Q$ of squares is then defined as the set of $4$-subsets of vertices of the form
\[
\{(g,00),(ag,01),(gb,10),(agb,11)\}.
\]
We see that this time the $Q$-neighbourhood of any vertex becomes naturally
in bijection with $A\times B$ without requiring any special properties of
$A=A^{-1}$
and $B=B^{-1}$.
In the present case the edge set $E_A$ of $\G_A$ becomes the set of pairs
of the form $\{(g,00),(ag,01)\}$ and of the form $\{(g,10),(ag,11)\}$, and the
edge set $E_B$ of $\G_B$ becomes the set of pairs $\{(g,00),(gb,10)\}$ and 
$\{(g,01),(gb,11)\}$. The graphs $\G_A$ and $\G_B$ are therefore both made up of two
connected components. As before, we may set $\G^\cup=(V,E_A\cup E_B)$, and
finally the graph $\G^\square$ over the vertex set $V$ puts an edge between
$(g,00)$ and $(agb,11)$ as well as between $(g,01)$ and $(agb,10)$ for all $g\in
G, a\in A, b\in B$. The two connected components $\G_0^\square$ and
$\G_1^\square$ have now become bipartite, the vertex set of $\G_0^\square$ being
$V_{00}\cup V_{11}$ and that of $\G_1^\square$ being $V_{01}\cup V_{10}$.

It is readily seen that Lemma~\ref{lem:lambda} also holds in the quadripartite
case with the eigenvalue analysis being similar. In the sequel we will not need
the bipartite structure of $\G_0^\square$ and $\G_1^\square$, and to lighten
notation we will just assume the complex to be bipartite and not necessarily
quadripartite. The quantum code construction presented in
Section~\ref{sec:construction} is however straightforwardly adapted to the quadripartite
case, and its analysis is essentially unchanged.

We note that the quadripartite version of the square left-right Cayley
complex appears in \cite{PK21}.

\subsection{Tanner codes}
A binary linear code of length $n$ is an $\f_2$-linear subspace of $\f_2^n$. For
sets $E$ of cardinality $|E|=n$, it will be convenient for us to identify $\f_2^n$
with $\f_2^E$, which we can think of as the space of functions from $E$ to
$\f_2$. Identication with $\f_2^n$ amounts to defining a one-to-one map between $E$ and
$[n]=\{1,2,\ldots ,n\}$, \textit{i.e.}\ a numbering of the elements of $E$.

Let $\G=(V,E)$ be a regular graph of degree $\Delta$, and for any vertex $v$
denote by $E(v)$ the set of edges incident to $v$. 
Assume an identification of $\f_2^{E(v)}$ with $\f_2^\Delta$ for every $v\in V$.
Let $x\in\f_2^{E}$ be a vector indexed by (or a function defined on) the set
$E$. Let us define the {\em local view} of $x$ at vertex $v$ as the subvector
$x_v=(x_e)_{e\in E(v)}$, \textit{i.e.}\ $x$ restricted to the edge-neighbourhood $E(v)$ of
$v$.

Let $C_0$ be a linear code of length $\Delta$, dimension $k_0=\rho_0\Delta$, and
minimum distance $d_0=\delta_0\Delta$.
We define the Tanner code \cite{T81} associated to $\G$ and $C_0$ as
\[
T(\G,C_0) =\{x\in\f_2^E : x_v\in C_0\;\text{for all}\;v\in V\}.
\]
In words, the Tanner code is the set of vectors over $E$ all of whose local
views lie in $C_0$.
By counting the number of linear equations satisfied by the Tanner code, we
obtain
\begin{equation}\label{eq:dimtanner}
\dim T(\G,C_0) \geq (2\rho_0-1)n.
\end{equation}
We also have the bound \cite{SS96,Gur} on the minimum distance $d$ of the Tanner
code:
\[
d\geq\delta_0 (\delta_0- \lambda(\G)/\Delta)n.
\]
Therefore, if $(\G_i)$ is a family of $\Delta$-regular expander graphs with
$\lambda(\G_i)\leq\lambda <d_0$, and if $\rho_0>1/2$, then the associated family of Tanner
codes has rate and minimum distance which are both $\Omega(n)$, meaning we have
an asymptotically good family of codes, as was first shown in \cite{SS96}.

\subsection{Quantum CSS codes}

A quantum CSS code is specific instance of a stabilizer code \cite{got97} that can be defined by two classical codes $\C_0$ and $\C_1$ in the ambient space $\f_2^n$, with the property that $\C_0^\perp \subset \C_1$ \cite{CS96,ste96}. If both codes are defined by their parity-check matrix, $\C_0 = \ker H_0, \C_1 =\ker H_1$, then the condition is equivalent to $H_0  H_1^T = 0$.
The resulting quantum code $\Q = (\C_0, \C_1)$ is the following subspace of
$(\mathbb{C}_2)^{\otimes n}$, the space of $n$ qubits:
\[ \Q := \mathrm{Span}\left\{ \sum_{z \in \C_1^\perp} |x+z\rangle \: : \: x\in \C_0 \right\},\]
where $\{ |x\rangle \: : \: x\in \f_2^n\}$ is the canonical basis of $(\mathbb{C}_2)^{\otimes n}$.

In practice, it is convenient to describe the code \emph{via} its \emph{generators}. A CSS code admits $X$-type generators which correspond to the rows of $H_1$ and and $Z$-type generators, corresponding to the rows of $H_0$.
The condition $\C_0^\perp \subset \C_1$ is simply that the rows of $H_0$ are orthogonal to the rows of $H_1$, where orthogonality is with respect to the standard inner product over $\f_2^n$.
A CSS code is called \emph{LDPC} is both $H_0$ and $H_1$ are sparse matrices, \textit{i.e.}\ each row and column has constant weight independent of the code length $n$. Equivalently, each generator acts nontrivially on a constant number of qubits, and each qubit is only involved in a constant number of generators. 

The dimension $k$ of the code counts the number of logical qubits and is given by 
\[ k = \text{dim} \, (\C_0/\C_1^\perp) = \text{dim} \, \C_0 + \text{dim} \, \C_1 - n.\]
Its minimum distance is $d = \min (d_X, d_Z)$ with
\[ d_X = \min_{w \in \C_0 \setminus \C_1^{\perp}} |w|, \quad d_Z = \min_{w \in \C_1 \setminus \C_0^\perp} |w|.\]
We denote the resulting code parameters by $\llbracket n,k,d\rrbracket$.
We say that a code family $(\Q_n)_n$ is \emph{asymptotically good} if its
parameters are of the form
\[ \llbracket n, k = \Theta(n),d = \Theta(n)\rrbracket.\]

\subsection{Tensor codes and dual tensor codes: robustness}
Let $A$ and $B$ be two sets of size $\Delta$. 
We define codes on the ambient space $\f_2^{A\times B}$
that we may think of the space of matrices whose rows (columns) are indexed by
$A$ (by $B$). If $C_A\subset \f_2^A$ and $C_B\subset\f_2^B$ are two linear codes,
we define the {\em tensor} (or product) code $C_A\otimes C_B$ as the space of
matrices $x$ such that for every $b\in B$ the column vector $(x_{ab})_{a\in A}$
belongs to $C_A$ and for every $a\in A$ the row vector $(x_{ab})_{b\in B}$
belongs to $C_B$. It is well known that $\dim(C_A\otimes
C_B)=\dim(C_A)\dim(C_B)$ and that the minimum distance of the tensor code is $d(C_A\otimes C_B)=d(C_A)d(C_B)$.

Consider the codes $C_A\otimes\f_2^B$ and $\f_2^A\otimes C_B$
consisting respectively of the space of matrices whose columns are codewords of
$C_A$ and whose rows are codewords of $C_B$. We may consider their sum
$C_A\otimes\f_2^B + \f_2^A\otimes C_B$ which is called a {\em dual
tensor} code, since it is the dual code of the tensor code
$C_A^\perp\otimes C_B^\perp=(C_A^\perp\otimes\f_2^B)\cap(\f_2^A\otimes
C_B^\perp)$. It is relatively straightforward to check that $d(C_A\otimes\f_2^B +
\f_2^A\otimes C_B)=\min(d(C_A),d(C_B))$.

\begin{defn}\label{def:robust}
Let $0\leq w\leq\Delta^2$. Let $C_A$ and $C_B$ be codes of length $\Delta$ with minimum distances $d_A$ and
$d_B$. We shall say that the dual tensor code $C=C_A\otimes\f_2^B +
\f_2^A\otimes C_B$ is $w$-{\em robust}, if for any codeword $x\in C$
of Hamming weight $|x|\leq w$, there exist $A'\subset A, B'\subset B$, $|A'|\leq
|x|/d_B$, $|B'|\leq |x|/d_A$, such that $x_{ab}=0$ whenever $a\notin A',
b\notin B'$.
\end{defn}

In words, $w$-robustness means that any dual tensor codeword of weight at
most $w$ is entirely supported within the union of a set of at most $|c|/d_A$
columns and a set of at most $|c|/d_B$ rows. In fact, the following proposition shows that any such codeword is the sum of a word of $C_A\otimes\f_2^B$ and of a word of $\f_2^A\otimes C_B$ supported on a few columns or rows.
\begin{prop}
Let $C_A$ and $C_B$ be codes of length $\Delta$ with minimum distances $d_A$ and
$d_B$, and suppose $C=C_A\otimes\f_2^B +
\f_2^A\otimes C_B$ is $w$-robust with $0<w<d_Ad_B$. Then for any codeword $x\in
C$ such that $|x|\leq w$,  there exist $A'\subset A, B'\subset B$, $|A'|\leq
|x|/d_B$, $|B'|\leq |x|/d_A$ and 
a decomposition $x = c + r$, with $c \in C_{A} \otimes \f_2^{B'}$ and $r \in \f_2^{A'}
\otimes C_B$. 
\end{prop}
\begin{proof}
To see this, apply the definition and write $x=r'+c'$, with
$r'_{ab}=c'_{ab}$ for any $(a,b)\in (A\setminus A')\times (B\setminus B')$.
The restrictions of $r'$ and $c'$ to $(A\setminus A')\times
(B\setminus B')$ both belong to the code obtained by tensoring $C_A'$ and
$C_B'$, the punctured codes deduced from $C_A$ and $C_B$ by throwing away
coordinates of $A'$ and $B'$. This code is the same as the punctured code
obtained from 
$C_A\otimes C_B$ by throwing away the coordinates $A'\times B \cup
A\times B'$. Therefore, there exists a tensor codeword of $C_A\otimes C_B=C_{A}
\otimes \f_2^{B}\cap \f_2^A\otimes C_B$ that
coincides with $c'=r'$ on $(A\setminus A')\times
(B\setminus B')$: adding this tensor codeword to both $c'$ and $r'$ yields the
required pair $r,c$ such that $x=r+c$.
\end{proof}

Our notion of robustness for the dual tensor code also implies a form of robustness for the corresponding tensor code.
\begin{prop}\label{prop:tensor}
Let $C_A$ and $C_B$ be codes of length $\Delta$ and minimum distances $d_A, d_B$
such that the dual tensor code $C_A\otimes\f_2^B + \f_2^A \otimes C_B$ is
$w$-robust with $w\leq d_Ad_B/2$. Then, any word $x$ close to both the column
and row code is also close to the tensor code: precisely,
if $d(x, C_A \otimes \f_2^{B})+ d(x, \f_2^{A} \otimes C_B) \leq w$ then 
\begin{align}\label{eqn:robust}
d(x, C_A \otimes C_B) \leq \frac{3}{2} \left( d(x, C_A \otimes \f_2^{B})+ d(x, \f_2^{A} \otimes C_B) \right).
\end{align}
\end{prop}

The conclusion of Proposition~\ref{prop:tensor} is very
close to a property called ``robustly testable'' in \cite{DEL21}. The difference
is that robustly testable tensor codes are required to satisfy \eqref{eqn:robust}
without any condition on its right hand side (equivalently,
with $w=2\Delta^2$), at the expense of allowing a looser constant than $3/2$.
We note that the
constant $3/2$ in Proposition~\ref{prop:tensor} is tight\footnote{For example, consider $C_A = C_B =
\mathrm{Span}(111100, 110011)$ and define $x = \left[ \begin{smallmatrix}
11&11&00\\11&
00&11\\10&00&00\\10&00&00\\01&00&00\\01&00&00\end{smallmatrix}\right]$. One
checks that $d(x, C_A \otimes \f_2^{B}) =  d(x, \f_2^{A} \otimes C_B)=4$ and $d(x, C_A \otimes C_B) =6$. }.

\begin{proof}
Let $x$ be an $A\times B$ matrix that we write as
$x =C + e_C = R + e_R$ with $C \in C_A \otimes \f_2^{B}$, $R \in \f_2^{A} \otimes C_B$ and 
\[ |e_C| = d(x, C_A \otimes \f_2^{B}), \quad  |e_R| = d(x,  \f_2^{A}\otimes C_B).\]
Let us suppose $|e_C|+|e_R|\leq w$.
From $C+R = e_C+e_R$, we obtain that $|C+R| \leq w$ and the robust testability of the dual product code implies that there exist $c \in C_A \otimes \f_2^{B}$, $r \in \f_2^{A} \otimes C_B$,
supported respectively, since $w\leq d_Ad_B/2$, on at most $d_B/2$ columns and at
most $d_A/2$ rows, and such that $C+R=c+r$. Since $c$ is supported by at most
$d_B/2$ columns, we have that $|c+r|\geq |c|$, and similarly $|c+r|\geq |r|$.

Note in particular that $R+r = C + c$ is a codeword of the tensor code $C_A \otimes C_B$.
One can therefore write $x = (C+c) + c+e_C$, from which we have:
\[
d(x,C_A\otimes C_B) \leq |c+e_C|\leq |c|+|e_C|\leq
|c+r|+|e_C|=|e_C+e_R|+|e_C|\leq 2|e_C|+|e_R|.
\]
Writing $x=(R+r)+r+e_R$, we similarly get $d(x,C_A\otimes C_B) \leq
2|e_R|+|e_C|$, and adding the two inequalities we conclude that
\[
d(x,C_A\otimes C_B) \leq \frac{ 3}{2}  \left(d(x, C_A \otimes \f_2^B) +
d(x,  \f_2^{A}\otimes C_B)\right). \qedhere
\]
\end{proof}

If $C_A\subset \f_2^A$ is a code and $A'\subseteq A$, let us denote by
$C_{A'}\subset\f_2^{A'}$ the {\em punctured code} consisting of all subvectors
$(c_a)_{a\in A'}$ of all codewords $(c_a)_{a\in A}$ of $C_A$. Similarly, for a
code $C_B$ we denote by $C_{B'}$ the punctured code on $B'$ for $B'\subseteq B$.

Let us introduce the following twist on the above definition of
robustness for dual tensor codes, which allows us to boost its potential.

\begin{defn}
Let $C_A\subset \f_2^A$ and $C_B\subset \f_2^B$. 
For integers $w,p$,
let us say that the dual tensor code $C_{A}\otimes\f_2^{B}+\f_2^{A}\otimes
C_{B}$ is $w$-robust with $p$-resistance to puncturing, 
if
for any  $A'\subset A$ and $B'\subset B$ such that $|A'|=|B'|=\Delta-w'$, with
$w'\leq p$, the dual tensor code
\[
C_{A'}\otimes\f_2^{B'}+\f_2^{A'}\otimes C_{B'}
\]
is $w$-robust.
\end{defn}

We shall need the following result on the robustness of random dual tensor
codes. 

\begin{restatable}{theo}{Probust}
\label{thm:P-robust}
Let $0<\rho_A < 1$ and $0<\rho_B<1$. Let $0<\varepsilon<1/2$ and $1/2+\eps<\gamma <1$. Let $C_A$ be a random code obtained from a
random uniform $\rho_A\Delta\times\Delta$ generator matrix, and let $C_B$ be a
random code obtained from a random uniform $(1-\rho_B)\Delta\times\Delta$
parity-check matrix. 
With probability tending to~$1$ when $\Delta$ goes to
infinity, the dual tensor code
\[
C_{A}\otimes\f_2^{B}+\f_2^{A}\otimes C_{B}
\]
is $\Delta^{3/2-\varepsilon}$-robust with $\Delta^\gamma$-resistance to puncturing.
\end{restatable}

Except for the fact that we allow a larger robustness parameter, (namely
$\Delta^{1/2-\varepsilon}$), this result is essentially in \cite{PK21}.
We provide a proof in the appendix, which closely follows the approach of
\cite{PK21}.

\section{Asymptotically good quantum Tanner codes}
\label{sec:main}

\subsection{The construction}\label{sec:construction}

Let $G$ belong to an infinite family of groups with size $|G| \to \infty$, together with two sets of generators $A=A^{-1}$ and $B=B^{-1}$ of fixed cardinality $\Delta$. We form the associated family of left-right Cayley complexes $X$, and will assume throughout that they satisfy the TNC condition \eqref{eqn:TNC}.

Recall that the square-neighbourhood (or link) $Q(v)$ of a vertex $v = (g,i) \in
G \times \{0,1\}$ is the set of squares incident to $v$. To lighten notation, let
us write, for $a\in A,b\in B$, $av=(ag,1-i), vb=(gb,1-i), avb=(agb,i)$, so that $\{v,av,vb,avb\}$ is a
square incident to $v$. 
The TNC condition implies in particular that the map
\begin{eqnarray*}
A\times B &\to& Q(v)\\
(a,b) & \mapsto & \{v,av,vb,avb\}
\end{eqnarray*}
is one-to-one, and let $\phi_v$ denote the inverse of this map.
A crucial consequence of the TNC condition is that for two adjacent vertices
$v,v'\in V$, the intersection $Q(v)\cap Q(v')$ has cardinality $\Delta$, and
more specifically, 
\begin{eqnarray*}
\phi_v(Q(v)\cap Q(vb))=A\times\{b\}&\quad&\phi_{vb}(Q(v)\cap
Q(vb))=A\times\{b^{-1}\}\\
\phi_v(Q(v)\cap Q(av))=\{a\}\times B&\quad&\phi_{av}(Q(v)\cap
Q(av))=\{a^{-1}\}\times B.
\end{eqnarray*}

For a vector $x\in\f_2^Q$, we may define its local view $x_v$ as the restriction
$(x_q)_{q\in Q(v)}$
of $x$ to the link $Q(v)$ of $v$.
Through the indexation map $\phi_v$, the local view at any vertex $v$ can be seen as an $A\times B$
array, 
all of its rows being shared by the local views of the $A$-neighbours of
$v$, and all of its columns being shared by the local views of its
$B$-neighbours. The situation is illustrated on Figure~\ref{fig:grid}.

\medskip

\begin{figure}[h]
\begin{center}
\begin{tikzpicture}
\draw (0,0) rectangle (3,3);
\draw[fill=black!10] (1,0) rectangle (1.5,3);
\draw[step=0.5cm] (0,0) grid (3,3);

\draw (4,0) rectangle (7,3);
\draw[fill=black!10] (6,0) rectangle (6.5,3);
\draw[step=0.5cm] (4,0) grid (7,3);

\node at (-0.7,1.5) {$Q(v)$};
\node at (7.8,1.5) {$Q(vb)$};
\node at (1.25,-0.4) {$b$};
\node at (6.35,-0.35) {$b^{-1}$};
\end{tikzpicture}
\end{center}
\caption{The structure of $Q(v)$ and $Q(vb)$ at neighbouring vertices $v$ and
$vb$. Both square-neighbourhoods are isomorphic to $A\times B$ and share a
common column, which is column $b$ for the $A\times B$ grid $Q(v)$ and column
$b^{-1}$ for the grid $Q(vb)$. \label{fig:grid}}
\end{figure}
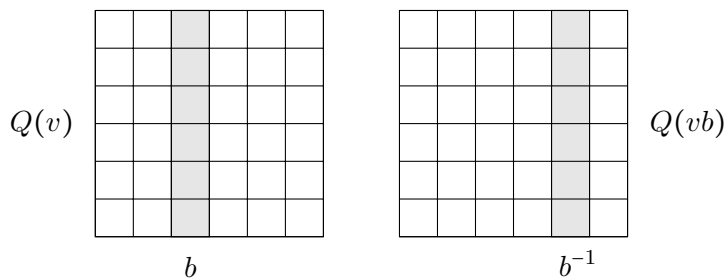

We define a quantum CSS code $\Q = (\C_0, \C_1)$ by associating qubits to the squares of the complex $X$ and enforcing constraints on the local view of each vertex. The idea is to associate $Z$-type generators with vertices of $V_0$ and $X$-type generators with $V_1$. Since their support can only overlap on a column or on a row, it suffices to impose orthogonality constraints between the two types of generators on each row and column to obtain a CSS code satisfying $\C_0^\perp \subset \C_1$. 
More precisely, we choose two codes $C_A\subset\f_2^A$ and $C_B\subset\f_2^B$ and
define the two tensor codes $C_0 = C_A \otimes C_B$ and $C_1 = C_A^\perp \otimes
C_B^\perp$. To define the two sets of generators, 
choose a basis $\beta_0$ of $C_0$ and a basis $\beta_1$ of $C_1$. For every
vertex $v\in V_0$ we define $\dim C_0$  generators $x\in\f_2^Q$ of type $Z$ by
requiring that their local views $x_v$ at $v$ be equal (through $\phi_v$) to a
basis element in $\beta_0$ and that $x_q=0$ for all $q\notin Q(v)$.
Similarly, we define $|V_1|\dim C_1$ generators of type $X$ by imposing a local
view equal to a basis element of $\beta_1$ on a single vertex $v\in V_1$ and
requiring $x_q=0$ for all values outside the neighbourhood $Q(v)$ of $v$.
We see that $X$-generators and $Z$-generators are orthogonal by design.

Equivalently, the code $\C_0$ ($\C_1$) orthogonal to all $Z$-generators
($X$-generators) is defined as the
set of vectors $x\in\f_2^Q$ such that $x_v$ is in $C_0^\perp$ (in $C_1^\perp$).
In other words, the quantum code $\Q = (\C_0, \C_1)$ is the pair of Tanner codes
\begin{equation}
\label{eq:tanner}
\C_0=T(\G_0^\square,C_0^\perp), \qquad \C_1=T(\G_1^\square,C_1^\perp).
\end{equation}

To have the same number of $X$ and $Z$-type generators, we shall set $\rho=\dim
C_A/\Delta$ and require $\dim C_B=\Delta -\dim C_A$: consequently we shall have
$\dim C_0=\dim C_1=\rho(1-\rho)\Delta^2$.

It is immediate that this code is LDPC since all its generators have weight at most $\Delta^2$, which is a constant, and any qubit is involved in at most $4 \rho(1-\rho) \Delta^2 \leq \Delta^2$ generators\footnote{By comparison, the quantum code of \cite{PK21} admits generators of weight $O(\Delta)$.}. 
The code length is the number of squares in the complex, $n = \Delta^2 |G|/2$.
While computing the exact dimension $k$ of the code would require to check for
possible dependencies among generators, it is straightforward to get a lower
bound by counting the number of generators, namely $|V_0|\dim C_0+|V_1|\dim C_1
=2 \rho (1-\rho) \Delta^2 |G|$.
This yields the following bound for the rate of the quantum code:
\begin{equation}\label{eq:rate}
 \frac 1n\dim\Q \geq (2\rho-1)^2,
\end{equation}
with equality when all the generators are independent. In particular, the rate is $>0$ for any value of $\rho \ne 1/2$.

Most of the rest of this section is devoted to establishing a linear lower bound for the
minimum distance of this quantum code, provided that the Cayley graphs
associated to $X$ are sufficiently expanding and the dual tensor codes
$C_0^\perp$ and $C_1^\perp$ are sufficiently robust. This in turn requires
$\Delta$ to be large enough and $G$ to be non Abelian (since the Cayley graph of
an Abelian group cannot be an expander if the degree is constant). But the
quantum Tanner code construction is more general, and even small examples where
$C_A = C_B^\perp$ is a small Hamming code could display good performance. For
instance, one could consider $G = \mathbbm{Z}/8\mathbbm{Z} \times
\mathbbm{Z}/2\mathbbm{Z}$, $A = \{ (a,0) \: : \: a\ne 0\}$, $B = \{(b,1) \: : \:
b\ne 1\}$ and $C_A$, $C_B$ to be the $[7,4,3]$ Hamming code and its dual
$[7,3,4]$.
The associated quantum Tanner code has length $392$, dimension at least $8$, and generators of weight $12$.

\subsection{Proof of Theorem \ref{thm:main}}

Let us recall our main technical theorem:

\main*

\begin{proof}

We consider a left-right Cayley complex $X$ from Section~\ref{subsec:LRCayley}
over a group $G$, satisfying the TNC condition and such that
$\Cay(G,A)$ and $\Cay(G,B)$ are Ramanujan graphs. 

Let $x \in \C_1$ be a codeword of weight $|x| <  \delta n/4\Delta^{3/2+\eps}$. It induces a subgraph $\G_{1,x}^\square$ of $\G_1^\square$ with vertex set $S \subset V_1$. 
The local view for any $v \in S$ corresponds to a codeword of $C_1^\perp =
C_A\otimes\f_2^B +\f_2^A\otimes C_B$. This code is $w$-robust so codewords of weight less than $w = \Delta^{3/2-\eps}$ have a support restricted to a small number $\leq w/(\delta \Delta)$ of rows and columns: in particular, the local view restricted to any column (or row) is at distance at most $\Delta^{1/2-\eps}/\delta$ from a codeword of $C_A$ (or $C_B$). 

Let us call \emph{normal vertices} the vertices of $S$ with degree less than $\Delta^{3/2-\eps}$ in $\G_{1,x}^\square$, and define $S_e$, the set of \emph{exceptional vertices} with degree greater than $\Delta^{3/2- \eps}$. Expansion in $\G_1^\square$ ensures that the set of exceptional vertices is small compared to $S$.

\begin{claim}\label{claim:Se}
The set of exceptional vertices has size 
\begin{align}\label{eqn:Se}
|S_e| \leq \frac{64}{\Delta^{1-2\eps}} |S|.
\end{align}
\end{claim}
\begin{proof}
The degree of each vertex $v \in S$ in the subgraph $\G_{1,x}^\square$ is at
least $\delta \Delta$ since the local view must correspond to a codeword of
$C_1^\perp = C_A\otimes\f_2^B +\f_2^A \otimes C_B$. It implies that 
\begin{equation}
\label{eq:|S|}
|S| \leq
\frac{2|x|}{\delta \Delta} \leq \frac{1}{4\Delta^{1/2+\eps}}|V_1|.
\end{equation}

Applying the Expander Mixing Lemma~\ref{lem:mixing} to $E(S_e,S)$ in
$\G_1^\square$, regardless of whether $\G_1^\square$ is bipartite or not, we obtain from
Lemma~\ref{lem:lambda} that:
\begin{align*}
 |E(S_e,S) | &\leq \frac{\Delta^2}{|V_1|}|S_e||S| +4\Delta \sqrt{|S_e| |S|}\\
             &\leq \frac 12\Delta^{3/2-\eps}|S_e| + 4\Delta \sqrt{|S_e| |S|}
\end{align*}
By definition of $S_e$, it also holds that $|E(S_e, S)| \geq \Delta^{3/2- \eps}|S_e|$.
Combining both inequalities, we obtain $\Delta^{1/2-\eps} \sqrt{|S_e|} \leq 8 \sqrt{|S|}$ and the claim follows.
\end{proof}

The support of the local view of any normal vertex of $S$ decomposes into a small
number of rows and columns, which are shared with vertices in $V_0$. We now
introduce the set $T$ of vertices of $V_0$ whose local views share with a
normal vertex of $S$ either a row or a column of large weight. Formally:

\paragraph{Defining the subset $T\subset V_0$.}
The vector $x$, viewed as a set of squares, defines a subset $E_x$ of edges of
$\G^\cup$, namely the edges incident to a square in $x$. Let us say that an edge
of $E_x$ is {\em heavy,} if it is incident to at least $\delta \Delta- \Delta^{1/2-\eps}/\delta$
squares of $x$. Let $T$ be the set of vertices of $V_0$ that are connected to
(at least) one {\em normal} vertex of $S$
{\em through a heavy edge}. Let us keep in mind that the local view of a normal
vertex of $S$ is supported by at most $\Delta^{1/2-\eps}/\delta$
rows and at most $\Delta^{1/2-\eps}/\delta$ columns, so a heavy edge between a
normal
vertex of $S$ and a vertex of $T$ corresponds to either a row or a column shared
by the two local views, which is at distance at most $\Delta^{1/2-\eps}/\delta$ from a nonzero codeword of $C_A$ (or $C_B$). 

\begin{claim}\label{claim:degreeT}
The degree in $E_x$ of any vertex of $T$ is at least $\delta\Delta- \Delta^{1/2-\eps}/\delta$.
\end{claim}

\begin{proof}
This follows simply from the fact that a row of weight $w$ in a local view is
incident to $w$ columns.
\end{proof}

\begin{claim}
For $\Delta$ large enough, the size of the set $T$ satisfies:
\begin{equation}\label{eqn:T}
|T|\leq \frac{64}{\delta^2\Delta}|S|.
\end{equation}
\end{claim}

\begin{proof}
From Claim~\ref{claim:degreeT} we have the following lower bound on the number
of edges of $\G^\cup$ between $S$ and $T$,
\[ |E(S,T)| \geq \delta \Delta \Big(1-  \frac{1}{\delta^2 \Delta^{1/2+\eps}}\Big) |T|.\]
Applying the Expander mixing Lemma~\ref{lem:mixing} to $E(S,T)$, we have, from
Lemma~\ref{lem:lambda},
\[
|E(S,T)|\leq \frac{4\Delta}{|V|}|S||T| + 4\sqrt{\Delta}\sqrt{|S||T|}.
\]
Recalling \eqref{eq:|S|}, we have $|S| \leq
|V|/8\Delta^{1/2+\eps}$, we have therefore
\[
|T|\delta\Delta\left(1-\frac{1}{\delta^2\Delta^{1/2+\eps}}-\frac{1}{2\delta\Delta^{1/2+\eps}}
\right)\leq 4\sqrt{\Delta}\sqrt{|S||T|}
\]
which implies 
\[
\delta\Delta\sqrt{|T|}\leq 8\sqrt{\Delta}\sqrt{|S|}
\]
for $\Delta$ large enough, from which the claim follows.
\end{proof}

From this last claim we infer that a typical vertex in $T$ must be adjacent to a large number of vertices in $S$, 
linear in $\Delta$, which means that its local view should consist of many
columns (or rows) containing almost undisturbed codewords of $C_A$ (or $C_B$). The robustness of $C_A \otimes C_B$ then implies that the local view of such a vertex must be close to a codeword of the tensor code, and that adding the corresponding word will decrease the weight of $x$. 
We now detail the argument.

Define $\bar{d}_T$ to be the average (over $T$) number of heavy edges incident
to a 
vertex of $T$. The number of edges between $T$ and $S \setminus S_e$ is $|E(T,
S\setminus S_e)| = \bar{d}_T |T| \geq |S| - |S_e|$. Applying~\eqref{eqn:Se}
and~\eqref{eqn:T}, we have
\[ \bar{d}_T \geq \frac{|S|-|S_e|}{|T|}\geq
\frac{\delta^2}{64}\left(1-\frac{64}{\Delta^{1-2\eps}}\right)\Delta=: 2\alpha \Delta. \]

Let $\eta$ be the fraction of vertices of $T$ with degree greater than $\alpha \Delta$.
Since the maximum degree of a vertex in $\G^{\cup}$ is $2\Delta$, it holds that
\[ 2\alpha \Delta \leq \bar{d}_T  \leq 2 \Delta \eta + (1-\eta)\alpha \Delta\]
and $\eta \geq \alpha/(2-\alpha) \geq \alpha/2$. We have just shown:

\begin{claim}\label{claim:heavy}
At least a fraction $\alpha/2$ of vertices of $T$ are incident to at
least $\alpha\Delta$ heavy edges.
\end{claim}

We will need to single out a vertex $v$ of $T$ whose existence is guaranteed by
claim~\ref{claim:heavy}, but we also need this vertex to not be incident to too
many exceptional vertices of $S$. To this end we estimate the total number of
edges of $\G^\cup$ between $T$ and $S_e$.

From the Expander mixing Lemma~\ref{lem:mixing}, 
\begin{align}
E(S_e,T)&\leq \frac{4\Delta}{|V|}|S_e||T| +
4\sqrt{\Delta}\sqrt{|T||S_e|}\nonumber\\
&\leq \frac{256\Delta^{2\eps}}{|V|}|T||S| +
32\Delta^{\eps}\sqrt{|T||S|}\nonumber\\
&\leq \frac{32}{\Delta^{1/2-\eps}}|T|+32\Delta^{\eps}\sqrt{|T||S|}\label{eq:32}
\end{align}
by first applying~\eqref{eqn:Se}, and then~\eqref{eq:|S|}. Now every vertex of
$S\setminus S_e$ is, by definition of $T$, adjacent to a vertex of $T$ in
$\G^\cup$. Since the degree of $\G^\cup$ is $2\Delta$ we get that $|S|-|S_e|\leq
2\Delta|T|$. From \eqref{eqn:T}, we have that, for $\Delta$ large enough,
$|S_e|\leq |S|/2$, whence $|S|\leq 4\Delta |T|$. From \eqref{eq:32} we therefore
obtain
\[
E(S_e,T)\leq \beta\Delta^{1/2+\eps}|T|
\]
with $\beta=64+32/\Delta$. We therefore have that at most an $\alpha/4$
proportion of vertices of $T$ are adjacent to more than $\frac
4\alpha\beta\Delta^{1/2+\eps}$ vertices of $S_e$. Summarising, we have shown,
together with claim~\ref{claim:heavy}

\begin{claim}\label{claim:alpha/4}
At least a fraction $\alpha/4$ of vertices of $T$ 
\begin{itemize}
\item are incident to at
least $\alpha\Delta$ heavy edges.
\item are adjacent to at most $d_1=\frac{4\beta}{\alpha}\Delta^{1/2+\eps}$ vertices of $S_e$. 
\end{itemize}
\end{claim}

Pick any vertex $v$ whose existence is guaranteed by claim~\ref{claim:alpha/4}.
The local view at $v$ is illustrated on Figure~\ref{fig:local}.

\begin{figure}[h!]
\begin{center}
\begin{tikzpicture}
\draw (0,0) rectangle (3,3); 
\draw[fill=black!15] (0.2,0) -- (0.2,2.4) -- (0,2.4) -- (0,2.7) -- (0.2,2.7) --
(0.2,3) -- (0.5,3) -- (0.5,2.7) -- (3,2.7) --
(3,2.4) -- (0.5,2.4) -- (0.5,0) -- (0.2,0);
\node at (-0.5,1.5) {$A$};
\node at (1.5,3.5) {$B$};
\end{tikzpicture}
\end{center}
\caption{The local view at $v$ from Claim~\ref{claim:alpha/4}. There are at most
$d_1\leq \Delta^\gamma$ rows
and $d_1$ columns (shaded) that are shared with local
views from exceptional vertices of $S$. All other rows and columns are shared
with normal vertices of $S$, and differ from a codeword of $C_B$ or $C_A$ on at
most $\Delta^{1/2-\eps}/\delta$ coordinates. \label{fig:local}}
\end{figure}
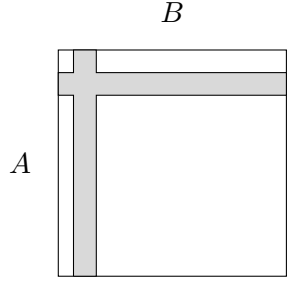

We may therefore define two subsets $A'\subset A$, $B'\subset B$, with
$|A'|=|B'|=\Delta-d_1$, and such that in the view of $v$, all rows and columns indexed by
$A'$ and $B'$ are normal, \textit{i.e.}\ are not shared with the local views of vertices
in $S_e$.
Observe that for $\Delta$ large enough
we have $d_1\leq\Delta^\gamma$, since we have imposed $\gamma>1/2+\eps$. The
theorem's hypothesis on resistance to puncturing implies therefore that
the punctured code $C_{A'} \otimes \f_2^{B'} + \f_2^{A'} \otimes C_{B'}$ is
$\Delta^{3/2-\eps/2}$-robust.

Denote by $x_v\in\f_2^{A\times B}$ the local view of $x$ at vertex $v$, 
and by $x_v'$ its restriction to coordinates in $A'\times B'$.
Proposition~\ref{prop:tensor}, applied for large enough $\Delta$, states that there exists a codeword $c' \in C_{A'}
\otimes C_{B'}$ of the punctured code $C_{A'} \otimes C_{B'}$ close to 
$x_v'$: specifically,
\[ d(x_v' , c') \leq  \frac{3}{2}  \left(d(x_v', C_{A'} \otimes \f_2^{B'}) +
d(x_v',  \f_2^{A'}\otimes C_{B'})\right) \leq 3  \Delta \frac{\Delta^{1/2-\eps}}{\delta}\]
since each column (or row) of $x_v$ is at distance at most
$\frac{\Delta^{1/2-\eps}}{\delta}$ from $C_A$ (or $C_B$), and there are not more
than $\Delta$ columns (rows) altogether.

Since $d_1<\delta\Delta$, there is a unique tensor codeword $c\in C_A \otimes
C_B$ that coincides with $c'$ on $A'\times B'$, and we have, since $|(A\times B)
\setminus (A'\times B')| \leq 2 d_1 \Delta$,
\[ d(x_v, c) \leq d(x_v', c') + 2 d_1 \Delta.\]
This last upper bound scales like $\Delta^{3/2+\eps}$, but the first point of
Claim~\ref{claim:alpha/4} implies that the weight of $x_v$ is a quantity that
scales like $\Delta^2$, so for large enough values of $\Delta$ it must be that
$d(x_v,c)=|x_v+c|<|x_v|$ (implying in particular that $c$ is non-zero).
Defining the vector $y\in\f_2^Q$ to coincide with $c$ on the $Q$-neighbourhood
of $v$ and to be zero elsewhere, we have just shown that
$|x|-|x+y|=|x_v|-|x_v+c|$ is positive, which concludes the proof of the theorem.
\end{proof}

\begin{rem}\label{rem:a}
The last part of the proof of Theorem~\ref{thm:main} shows that
when $\Delta$ is large enough, not only do we have $|x|-|x+y| >0$ for $y$ thus
constructed, but there is furthermore a constant $a$ such that we have
$|x|-|x+y| >a\Delta^2.$ 
\end{rem}

\begin{rem}\label{rem:negative}
Theorem~\ref{thm:main} stays valid for negative values of $\eps$, \textit{i.e.}
when $\eps\in (-1/2,0)$. In this case however, there is no need for any resistance
to puncturing, because the exceptional rows and columns at the local view of $v$
from Claim~\ref{claim:alpha/4} will number $o(\Delta^{3/2})$ and contribute a
negligible amount to the distance between $x_v$ and the tensor codeword $c$.
\end{rem}

\subsection{Consequences of Theorem~\ref{thm:main}: asymptotically good quantum LDPC codes}
\label{sec:goodLDPC}
The following theorem is a direct consequence of Theorem~\ref{thm:main}.

\begin{theo}\label{thm:123}
Fix $\rho\in(0,1/2)$, $\eps \in (0,1/2)$, $\gamma\in (1/2+\eps,1)$ and $\delta>0$. 
If $\Delta$ is large enough and $C_A$ and $C_B$ are codes of length $\Delta$
such that
\begin{enumerate}
\item $0<\dim C_A\leq\rho\Delta$ and $\dim C_B=\Delta-\dim C_A$,
\item the minimum distances of $C_A,C_B,C_A^\perp,C_B^\perp$ are all
$\geq\delta\Delta$,
\item \label{cond3} both dual tensor codes $C_0^\perp=(C_A\otimes C_B)^\perp$ and
$C_1^\perp=(C_A^\perp\otimes C_B^\perp)^\perp$ are $\Delta^{3/2-\eps/2}$-robust
with $\Delta^\gamma$-resistance to puncturing,
\end{enumerate}
then the quantum code $\Q=(\C_0,\C_1)$ defined in \eqref{eq:tanner} has 
parameters
\[
\llbracket n, k\geq (1-2\rho)^2n, d\geq \frac{\delta}{4\Delta^{3/2+\eps}}n
\rrbracket .
\]
\end{theo}

\begin{proof}
The statement on the dimension $k$ is formula~\eqref{eq:rate}.
Recall that the minimum distance of the quantum code $\Q$ is the smallest weight
of a vector that is either in $\C_0\setminus \C_1^\perp$ or in $\C_1\setminus
\C_0^\perp$. Let $x\in\C_1$ be a non-zero vector of weight $<\delta
n/4\Delta^{3/2+\eps}$.
The vector $y$ guaranteed by Theorem~\ref{thm:main} is in
$\C_0^\perp\subset\C_1$, and since $|x+y|<|x|$, we may apply repeatedly
Theorem~\ref{thm:main} to create a sequence of vectors of $\C_1$ of decreasing
weights, and ultimately obtain that $x$ must be a sum of vectors of
$\C_0^\perp$. The same argument applies symmetrically to the weight of vectors
of $\C_0\setminus\C_1^\perp$, since we have imposed that $C_A^\perp$ and
$C_B^\perp$ also satisfy the hypotheses of Theorem~\ref{thm:main}.
\end{proof}

To obtain asymptotically good families of quantum codes it remains to show that
codes $C_A,C_B$ satisfying the conditions 1,2,3 of Theorem~\ref{thm:123}
exist. This is achieved through the random choice result of
Theorem~\ref{thm:P-robust}.

Specifically, we have:

\begin{theo}\label{thm:goodcodes}
Fix $\rho\in(0,1/2)$, $\eps \in (0,1/2)$, $\gamma\in (1/2+\eps,1)$ and
$\delta >0$ such that $-\delta\log_2\delta-(1-\delta)\log_2(1-\delta)<\rho$.
Fix some large enough $\Delta$ and let $r=\lfloor\rho\Delta\rfloor$. Let $C_A$
be the random code defined by a random uniform $r\times\Delta$ generator matrix
and let $C_B$ be the random code defined by a random uniform $r\times\Delta$
parity-check matrix. With non-zero probability $C_A$ and $C_B$ satisfy
conditions 1,2,3 of Theorem~\ref{thm:123} yielding an infinite family of quantum
codes of parameters
 \[
\llbracket n, k\geq (1-2\rho)^2n, d\geq \frac{\delta}{4\Delta^{3/2+\eps}}n
\rrbracket .
\]
\end{theo}

\begin{proof}
It is well-known that for $r<\Delta$, a random uniform $r\times\Delta$ matrix has rank $r$ with
probability at least $1-1/2^{\Delta-r}$ so condition $1$ will hold with
probability tending to $1$ when $\Delta$ goes to infinity.

The value $\delta$ has been chosen to be below the Gilbert-Varshamov bound for
the imposed rates of all $4$ codes $C_A,C_B,C_A^\perp,C_B^\perp$, so all their
minimum distances will be $\geq\delta\Delta$ with probability tending to $1$
when $\Delta$ goes to infinity.

Finally, both the pair $(C_A,C_B)$ and the pair $(C_A^\perp,C_B^\perp)$ are
chosen according to the same probability distribution, namely one random uniform
generator matrix and one random uniform parity-check matrix, therefore, with
probability tending to $1$ when $\Delta$ goes to infinity, both dual
tensor codes $C_0^\perp$ and $C_1^\perp$ will satisfy the conclusion of
Theorem~\ref{thm:P-robust} in other words condition \ref{cond3}.
\end{proof}

\paragraph{An upper bound on the minimum distance of $\Q$.} 
To evaluate the tightness of the lower bound on the minimum distance $d$ of the quantum code $\Q$
given by Theorem~\ref{thm:123}, we now derive the following upper bound on $d$.
\begin{prop}\label{prop:n/Delta}
The minimum distance $d$ of the quantum code $\Q$
is not more than $n/\Delta$.
\end{prop}

Recalling Remark~\ref{rem:negative}, we have that if we could
find the required dual tensor codes $C_0^\perp$ and $C_1^\perp$ that are both
$\Delta^{3/2-\eps}$-robust for $\eps\to -1/2$, we would obtain a lower bound on
the quantum code minimum distance that practically closes the gap with the upper
bound of Proposition~\ref{prop:n/Delta}.

To prove Proposition~\ref{prop:n/Delta} we shall show the existence of
small-weight ($\leq n/\Delta$) codewords of $\C_1\setminus \C_0^\perp$. We do this by constructing
many small-weight codewords of $\C_1$, that are too numerous to all be in
$\C_0^\perp$. Those codewords are best described when using the quadripartite
version of the Left-Right Cayley complex described at the end of
Section~\ref{subsec:LRCayley}. In this case the graph $\G_1^\square$ is
bipartite, with its vertex set $V_1$ split into the disjoint union of two sets $V_{10}$ and $V_{01}$.
We have that for every vertex $v\in V_1$, all squares in its local view that are
indexed by row $a$, are indexed by row $a^{-1}$ in all neighbouring local views.
Therefore, if for every vertex $v\in V_{10}$ we restrict its $Q$-neighbourhood
$Q(v)$ to row $a$ for a fixed $a$, and similarly restrict all $Q$-neighbourhoods 
of vertices of $V_{01}$ to row $a^{-1}$, we obtain a subgraph $\G_{1a}^\square$
of $\G_1^\square$ that is a copy of the double cover of the graph $\Cay(G,B)$.
Furthermore, the edge set of $\G_1^\square$ is the disjoint union of the edge
sets of the graphs $\G_{1a}^\square$ for $a$ ranging in $A$.

Now we see that any codeword of the Tanner code
$T(\G_{1a}^\square,C_B)$ yields a codeword of $\C_1$, where every non-zero local view at
any vertex of $V_{10}$ must be entirely supported by row $a$ (and row $a^{-1}$
for vertices of $V_{01}$) and coincide with a codeword of $C_B$. Any such
codeword has weight at most $n/\Delta$: expanding
$T(\G_{1a}^\square,C_B)$ to the whole index set $\G_1^\square$ by padding its
words with $0$s, we may define the disjoint direct sum
\[
\eL = \sum_{a\in A}T(\G_{1a}^\square,C_B).
\]
We claim that:
\begin{lemma}
\label{lem:eL}
$\dim\eL\cap\C_0^\perp \leq 2\dim C_A.\dim T(\G_{1a}^\square,C_B)$ for any given
$a\in A$.
In particular when $\dim C_A<\Delta/2$ we have  $\eL\not\subset\C_0^\perp$.
\end{lemma}
The second claim of Lemma~\ref{lem:eL} follows from the first, since clearly
$\dim\eL=\Delta\ell$ where $\ell$ is the common dimension of the $\Delta$
isomorphic Tanner codes $T(\G_{1a}^\square,C_B)$.
Lemma~\ref{lem:eL} implies therefore that at least
one non-zero Tanner codeword in some $T(\G_{1a}^\square,C_B)$ (which must be of weight
$\leq n/\Delta$) cannot be in $\C_0^\perp$ which proves
Proposition~\ref{prop:n/Delta}. It remains therefore to prove
Lemma~\ref{lem:eL}.
We will need the following easy fact on Tanner codes on bipartite graphs:

\begin{lemma}
\label{lem:tannereasy}
Let $\G=(W,E)$ be a regular bipartite graph on the vertex set $W=W_0\cup W_1$, and let
$C=T(\G,C_0)$ be a Tanner code on it for some inner code $C_0$. Let
$x=\sum_{w\in W}c_v$ be a vector in $\f_2^E$ such that every $c_w$ is a word
supported by the edge-neighbourhood $E(w)$ of $w$ and that coincides with a codeword of
$C_0$ on $E(w)$. Then if $x$ is a Tanner codeword of $C$, then so are the
partial sums $\sum_{w\in W_0}c_w$ and $\sum_{w\in W_1}c_w$.
\end{lemma}

\begin{proof}
We have that $x=\sum_{w\in W_0}x_w=\sum_{w\in W_1}x_w$ where $x_w$ is the local
view of $x$ at vertex $w$. Therefore $\sum_{w\in W}c_w+\sum_{w\in W_1}x_w=0$,
so that 
\[
\sum_{w\in W_0}c_w =\sum_{w\in W_1}c_w+x_w
\]
and therefore $\sum_{w\in W_0}c_w$ is a vector whose local views, both at
vertices of $W_0$ and at vertices of $W_1$, are all in $C_0$, so $\sum_{w\in W_0}c_w$ 
is a Tanner codeword. Similarly, so is $\sum_{w\in W_1}c_w$.
\end{proof}

\begin{proof}[Proof of Lemma~\ref{lem:eL}]
Let $e_1,\ldots e_k$ be a basis of $C_A$ corresponding to some information
set $a_1,\ldots ,a_k$, for $k=\dim C_A$. This means that $e_i$ has value $1$ on
index $a_i$ and has value $0$ on all indices $a_j, j\neq i$. 

Let $x\in\eL\cap\C_0^\perp$. Since $x\in\C_0^\perp$, it decomposes as
\begin{equation}\label{eq:tannerdecomp}
x = \sum_{i=1}^k\sum_{v\in V_{00}}e_i\otimes x_v^i + \sum_{i=1}^k\sum_{w\in
V_{11}}e_i\otimes x_w^i.
\end{equation}
where $x_v^i,x_w^i$ are all codewords of $C_B$. We have abused notation somewhat
by identifying the local view of $x$ at a vertex $v$ with its decomposition
expressed as a tensor codeword. But the point is that every term of the form 
$e_i\otimes x_v^i$ contributes to the sum \eqref{eq:tannerdecomp} a codeword of $C_B$
to the local view at $v$ on row $a_i$. Since we must see on column $a_i$ a
Tanner codeword of $T(\G_{1a_i}^\square,C_B)$ (because $x\in\eL$), we have from 
Lemma~\ref{lem:tannereasy} that the whole sum
\[
\sum_{v\in V_{00}}e_i\otimes x_v^i 
\]
must be equal, when restricted to row $a_i$, to a codeword of
$T(\G_{1a_i}^\square,C_B)$. Therefore, the left part of the
sum~\eqref{eq:tannerdecomp} must belong to a sum of $\dim C_A$ codes all isomorphic to
a $T(\G_{1a}^\square,C_B)$ Tanner code, and by a similar argument, so must the
right part of the sum, which concludes the proof.
\end{proof}

\subsection{Consequences of Theorem~\ref{thm:main}: classical locally testable codes}

Let $X$ be the left right Cayley complex of Theorem~\ref{thm:main}, 
and let $C_A = C_B$ be a code of length $\Delta$, rate $\rho \in (0,1)$ and distance
$\delta \Delta$ and suppose that the associated dual tensor code
$C_A\otimes\f_2^B+\f_2^A\otimes C_B$ satisfies
Theorem~\ref{thm:main}. 

 Let $C_0=C_A\otimes C_B$:
the Tanner code $\C = T(\G_0^\square, C_0)$ is precisely the locally testable
code defined of Dinur \textit{et al.}\ and is shown (\cite{DEL21}) to have parameters
\[ [n, k \geq (2\rho^2-1) n, d \geq \delta^2(\delta- \lambda/\Delta)n]\]
where $\lambda$ is the (common) second largest eigenvalue of the constituent
Ramanujan Cayley graphs.

Recall from \cite{DEL21} the tester for this code: given a word $x \in \f_2^Q$, pick a random vertex $v \in V_0$, accept if the local view $x_v$ belongs to the tensor code $C_0$ and reject if $x_v \notin C_0$.
Let $T \subset V_0$ be the set of vertices for which the local test rejects,
\begin{align*} 
 T = \{ v \in V_0 \: : \: x_v \notin C_0\}.
\end{align*}
The fraction of rejecting local tests is therefore $\zeta(x) := \frac{|T|}{|V_0|}$.

Recall that a code is said to be locally testable with $q$ queries and detection
probability $\kappa$, if the tester accesses at most $q$ bits from $x$, always
accepts when $x$ is a codeword, and otherwise satisfies 
\begin{equation}\label{eq:testable}
 \zeta(x) \geq \kappa \frac 1n d(x,\C).
\end{equation}
In the present case, the number of queries is $q=\Delta^2$, the size of the
$Q$-neighbourhood of a vertex. The goal is to establish \eqref{eq:testable} for
some constant $\kappa$.

The strategy to establish the local testability of the code $\C$ is to define a
decoder that is guaranteed to always find a codeword close to $x$ if $\zeta(x)$ (or $|T|$) is sufficiently small.
The difference with a classical decoder is that we make no assumption on how far
$x$ actually is from the code $\C$.

Theorem~\ref{thm:main} can be converted into such a decoding algorithm. Let $x
\in \f_2^Q$ be an initial vector. Our goal is to find a close enough codeword $c \in \C$. 
For each vertex $v \in V_0$, let $c_v \in C_0$ be the closest codeword to the local view $x_v$ (breaking ties arbitrarily), and let $e_v := x_v + c_v$ be the local corresponding error (of minimum weight). We slightly abuse notation here and write $c_v$ or $e_v$ for both the vectors in $\f_2^{Q(v)}$ and the vectors in $\f_2^{Q}$ coinciding with $e_v, c_v$ on $Q(v)$ and equal to zero elsewhere.
\paragraph{The decoder.}
The decoder starts by computing the list $(c_v)_{v\in V_0}$ from the
decompositions $x_v=c_v+e_v$.
Note that this local decoding can be achieved by brute-force if need be since
the local code has constant length $\Delta^2$. Note also that the list of these local
views is not necessarily equal to the list of local views of a codeword
$c\in\C$ (if it does then the decoder outputs $c$).
The decoder then computes what we may call
the \emph{mismatch} of the list $(c_v)$, and which is defined as $z = \sum_v c_v$.
If the weight $|z|$ is too large, namely $\geq\delta
n/4\Delta^{3/2+\eps}$, then the decoder will refuse to continue and output
``far from the code''.
Otherwise the decoder proceeds by looking for a vertex $v\in V_0$ on which it
will update the value of $c_v$, replacing it by $c_v'=c_v+y_v$ for some non-zero
$y_v\in C_0$, so as to decrease the weight of the new value $z+y_v$ of the
mismatch. Among all possible vertices $v$ and small codewords $y_v\in C_0$, let it
choose the one that maximizes $|z|-|z+y_v|$.

The decoder proceeds iteratively in this way until it has a list of local views
with a zero mismatch, corresponding therefore to the list of local views of a
global codeword $c'\in\C$ which it outputs.

\begin{claim}
If $|z| < \delta n/4 \Delta^{3/2+\eps}$, then the decoder always converges to
the zero mismatch and a codeword $c'$ of $\C$.
\end{claim}

\begin{proof}
The crucial observation is that $z$ is a codeword of $\C_1 = T(\G_1^\square, C_A \otimes \f_2^B + \f_2^A \otimes C_B)$, the Tanner code defined on the graph $\G_1^\square$ (with vertices in $V_1$), and local code given by the \emph{dual} tensor code. In the language of the quantum codes of the previous section, this is simply because $\sum_{v \in V_0} c_v$ is a sum of generators. 
Theorem~\ref{thm:main} asserts that there exists a generator $y_v \in C_0$ such
that $|z| - |z + y_v| > a \Delta^2$, for some constant $a>0$, independent
of~$n$ (and mentioned in Remark~\ref{rem:a}). 
\end{proof}

\begin{claim}
If $|z| < \delta n/4 \Delta^{3/2+\eps}$, then the distance from $x$ to $\C$ is bounded by
\begin{align}\label{eqn:ltc1}
d(x,\C) \leq  n\Big(1+\frac{1}{a}\Big) \zeta(x). 
\end{align}
\end{claim}

\begin{proof}
Let $(c_v')_{v\in V_0}$ be the list of local views of the output codeword $c'$.
Writing $x=c'+e'$, we have that the local views $e'_v$ of $e'$ satisfy
$x_v'=c_v'+e_v'$ and 
\[
d(x,c') =|e'|= \frac 12\sum_{v\in V_0}|e_v'|.
\]
Let $S$ be the list of vertices of $V_0$ whose output value $c_v'$ differs from
the original value $c_v$. Writing $x=c'+e'$, and remembering that $T$ is the set
of vertices $v$ for which the original $e_v$ is non-zero, we may bound from
above $|e'|$ by
\[ \sum_{v \in T} |e_v| + \sum_{v\in S} |e'_v| \leq \Delta^2(|T| +|S|).\]

We must have $|S| \leq |z|/(a \Delta^2)$. Furthermore,
since each bit of the word $x$ appears twice in the sum $\sum_{v \in V_0} x_v =
\sum_{v\in V_0} (c_v + e_v)$, we have that $\sum_{v\in V_0} c_v = \sum_{v \in
V_0} e_v = z$, from which we infer that $|z|\leq |T|\Delta^2$. Putting this
together, we obtain

\[ |e'|\leq \Delta^2 ( |T| + |T|/a) = \Delta^2 |V_0| \Big(1+\frac{1}{a}\Big) \zeta(x)\]
whence
\[ d(x,\C) \leq   n\Big(1+\frac{1}{a}\Big) \zeta(x). \qedhere \]
\end{proof}

Let us now consider a word $x \in \f_2^Q$ such that $|z| \geq \delta n/4
\Delta^{3/2+\eps}$, in which case the decoder gives us nothing. We nevertheless
have that, again using $|z|\leq\Delta^2|T|$, 
\begin{align}\label{eqn:ltc2}
\frac 1nd(x, \C) \leq 1 \leq \frac{4 \Delta^{3/2+\eps} |z|}{\delta n} \leq
\frac{4 \Delta^{7/2+\eps} |T|}{\delta n} = \frac{4 \Delta^{7/2+\eps}
|V_0|}{\delta n} \zeta(x) = \frac{8 \Delta^{3/2+\eps}}{\delta} \zeta(x). 
\end{align}

From \eqref{eqn:ltc1} and \eqref{eqn:ltc2}, we obtain that the Tanner code $\C = T(\G_0^\square, C_0)$ is $\kappa$-locally testable with $\Delta^2$ queries and
\[ \kappa = \min \Big( \frac{a}{a+1}, \frac{\delta}{8 \Delta^{3/2+\eps}} \Big).\]

We conclude this section with a word of comment on the choice of the small component
code $C_A=C_B$. To obtain an LTC we need it to satisfy the hypotheses of
Theorem~\ref{thm:main}. One way is to obtain this by random choice, namely by
applying Theorem~\ref{thm:P-robust}, as we did in Section~\ref{sec:goodLDPC}.
But since this time we have no need for the robustness of $C_A^\perp\otimes
C_B^\perp$, there are alternatives: in \cite{DEL21}, the
component codes that are used have better robustness than what is guaranteed by 
 Theorem~\ref{thm:P-robust}.

\newpage

\appendix 
\section{Appendix: proof of Theorem~\ref{thm:P-robust}}
The proof follows the blueprint of \cite{PK21}, though puncturing is handled a
little differently.

Let $C_A\subset\f_2^A$ and $C_B\subset\f_2^B$ and let $C$ be the dual tensor
code $C= C_A\otimes
\f_2^B+\f_2^A\otimes C_B$. Let $X\subset\f_2^{A\times B}$ be a
$\Delta\times \Delta$ matrix. 
Let $H_A$ be an $r_A\times\Delta$ parity-check matrix for the
column code $C_A$.
The $r_A\times\Delta$ matrix $H_AX$ is such that, for every $b\in B$, its $b$-th
column is the $H_A$-syndrome of the $b$-th column of $X$.
We note that $X\in C$ iff $H_AX$ is such that everyone of its $r_A$ rows is a
codeword of $C_B$.

\begin{lemma}
\label{lem:injective}
Let $V$ be a set of vectors of $\f_2^{A}$. Suppose that the union of the
supports in $A$ of the vectors of $V$ has size $t<d_A/2$, where $d_A$ is the minimum
distance of $C_A$. Then $\rk{\{H_Av : v\in V\}}=\rk{V}$, where $\rk{}$ denotes the rank
function.
\end{lemma}
\begin{proof}
The syndrome map $x\mapsto H_Ax^\intercal$ is injective on the Hamming ball of
radius $t$, therefore for a subset $W$ of vectors $x$ of $V$, $\sum_W H_Ax=0$ implies
$\sum_W x=0$. So if $W\subset V$ is a set of linearly independent vectors,
$\rk{\{H_Ax : x\in W\}}=|W|$.
\end{proof}

Let $d> 2\alpha \Delta_A$ be a lower bound on the minimum distance $d_A$ of $C_A$,
for some constant $\alpha$.

\begin{lemma}
\label{lem:rank}
Let $X$ be a $\Delta\times n$ matrix such that all of its columns are of
weight $\leq \sqrt{\Delta}/\log_2\Delta$ and such that the number of its
non-zero rows is at least $d/2$.
Then, for $\Delta$ large enough, $X$ has rank at least
$\alpha\sqrt{\Delta}\log\Delta$.
\end{lemma}

\begin{proof}
Let $x_1,x_2,\ldots ,x_m$ be a maximum sequence of non-zero columns of $X$ with the
property that $\supp(x_i)\not\subset\cup_{1\leq j<i}\,\supp(x_j)$. 
We clearly have that $x_1\ldots ,x_m$ are linearly independent, and
we have $m\geq d/2w$ where $w$ is the maximum weight of a vector $x_i$, otherwise
the union of the supports of all columns of $X$ would be less than $d/2$, contrary
to our hypothesis. From $w\leq \sqrt{\Delta}/\log_2\Delta$ we have the
result.
\end{proof}

\begin{corol}\label{cor:rank}
Let $X$ be a $\Delta\times n$ matrix such that all of its columns are of
weight $\leq \sqrt{\Delta}/\log_2\Delta$ and that has at least $d/2$ non-zero
rows. Then $H_AX$ has rank at least
$\alpha\sqrt{\Delta}\log\Delta$.
\end{corol}

\begin{proof}
Choose $\alpha\sqrt{\Delta}\log\Delta$ linearly independent columns of $X$, which is possible by
Lem\-ma~\ref{lem:rank}. Since the weight of these vectors is not more than $\leq
\sqrt{\Delta}/\log_2\Delta$, the sum of all the weights of these vectors is
less than $d/2$, therefore Lemma~\ref{lem:injective} applies and the claim
is proved.
\end{proof}

\begin{lemma}
\label{lem:randomC}
Let $\bC$ be a random linear code of length $n$, defined by an $r\times n$ uniform
random matrix $\bH$. Let $V$ be a set of linearly independent vectors of
$\f_2^n$. The probability that all vectors of $V$ fall into the random code $\bC$
is equal to $1/2^{r|V|}.$ 
\end{lemma}

\begin{proof}
The $\bH$-syndrome of a fixed non-zero vector $x$ is uniformly distributed in
$\f_2^r$ and the probability that it equals zero is therefore $1/2^r$. When the
fixed vectors $x\in V$ are independent in the sense of linear algebra, their
$\bH$-syndromes are independent in the sense of probability, hence the lemma.
\end{proof}

\begin{lemma}
\label{lem:randomCp}
Let $\bC$ be a random linear code of length $n$, defined by an $r\times n$ uniform
random matrix $\bH$. Let $\bC_p$ be the punctured code obtained from $\bC$ by
throwing away the first $p$ coordinates. Let $V$ be a set of linearly
independent vectors of $\f_2^{n-p}$ defined on the set of coordinates $\{p+1,\ldots n\}$.
Then the probability that all vectors of $V$ fall into the random code $\bC_p$
is at most $1/2^{(r-p)|V|}.$ 
\end{lemma}

\begin{proof}
Let $\bW$ be the random linear subspace of $\f_2^r$ generated by the first $p$ columns
of $\bH$. Let $\bH_p$ be the $r\times (n-p)$ matrix deduced from $\bH$ by
throwing away its first $p$ columns. A vector
$x\in\f_2^{n-p}$ is in the punctured code $\bC_p$ iff its $\bH_p$-syndrome
$\bH_px^\intercal$
belongs to $\bW$. Denote by $E_V$ the event whereby all vectors of $V$ fall into
the punctured code $\bC_p$. 
Let $W$ be a fixed subspace of $\f_2^r$ of dimension $w\leq
p$. 
Denote by $E_W$ the event whereby the first $p$ columns of $\bH$ generate $W$.
We have that the projections onto $\f_2^r/W$ of the columns of $\bH_p$ are uniformly distributed 
and independent in the sense of probability. Therefore, if the first $p$ columns
of $\bH$ generate the fixed subspace $W$, Lemma~\ref{lem:randomC}
applies and we have $P(E_V | E_W)=1/2^{(r-w)|V|}$. Now we have
\[
P(E_V) = \sum_WP(E_W)P(E_V|E_W)
\]
where the sum is over all subspaces $W$ of $\f_2^r$ of dimension at most $p$.
Hence,
\[
P(E_V)\leq\sum_WP(E_W)\frac{1}{2^{(r-p)|V|}} = \frac{1}{2^{(r-p)|V|}}. \qedhere
\]
\end{proof}

\begin{lemma}
\label{lem:XP}
Let $p,r_B$ be integers such that $0<p<r_B<\Delta$. Let $X_p$ be
a fixed $\Delta\times(\Delta-p)$ binary matrix such that all its columns
are of weight $\leq \sqrt{\Delta}/\log_2\Delta$ and that has at least $d/2$
non-zero rows.
Let $P$ be a subset of $B$ of cardinality $p$.
Let $\X_P$ be the subset of $\Delta\times\Delta$ matrices that, when leaving
out all columns indexed by $P$, are all equal to $X_p$.

Let $\bC$ be the random code defined by a uniform
random $r_B\times\Delta$ parity-check matrix $\bH$. Then the probability that
the dual tensor code 
$C_A\otimes\f_2^B+\f_2^A\otimes\bC$ contains at least one matrix in $\X_P$
is at most $1/2^{(r_B-p)\alpha\sqrt{\Delta}\log\Delta}$.
\end{lemma}

\begin{proof}
Identifying $B$ with $\{1,\ldots ,\Delta\}$, we may suppose without loss of
generality that $P$ is equal to the set
$\{1,\ldots ,p\}$ of the
first $p$ coordinates. Let $\bC_p$ be deduced from
$\bC$ by puncturing the first $p$ coordinates. If $\X_P$ contains one matrix in
the dual tensor code, then it must be that all rows of $H_AX_p$ fall into $\bC_p$,
where $H_A$ is the parity-check matrix of $C_A$.
By Corollary~\ref{cor:rank} we have that the rank of $H_AX_p$ is at least $\alpha\sqrt{\Delta}\log\Delta$.
We now apply Lemma~\ref{lem:randomCp} to conclude.
\end{proof}

Let $0<\rho<1$ and let $0<\delta<1$
be such that $\delta < h^{-1}(1-\rho)$, where
$h(x)=-x\log_2(x)-(1-x)\log_2(1-x)$ denotes the binary entropy function.
In other words we take $\delta$ to be less than the Gilbert-Varshamov bound on
the relative minimum distance for codes of rate $\rho$.
This means in particular that, with probability 
tending to $1$ when $\Delta$
goes to infinity (more precisely, behaving as $1-1/2^{\Omega(\Delta)}$), a random code of rate $\rho$ and length $\Delta$
has relative minimum distance at least $\delta$.

\begin{theo}\label{thm:robustness-dual}
Let $C_A$ be a fixed code of length $\Delta$ and of minimum distance at least
$\delta\Delta$. Let $\bC$ be a random code obtained from a uniform random
parity-check matrix of size $r\times\Delta$ with $r\geq\Delta(1-\rho)$. With
probability at least $1-1/2^{c\Delta^{3/2}}$, with $c$ depending only on $\rho$,
we have that: 

every matrix $X$ of weight at most
$\frac\delta 2\Delta^{3/2}/\log_2\Delta$ is 
\begin{itemize}
\item either supported by at most $\delta\Delta/2$ rows and at most
$\delta\Delta/2$ columns,
\item or not in the dual tensor code $C_A\otimes\f_2^B+\f_2^A\otimes\bC$.
\end{itemize}
\end{theo}

\begin{proof}
A matrix $X$ whose weight is thus upper bounded contains at most
$p=\delta\Delta/2$ columns of weight $>\sqrt{\Delta}/\log\Delta$. Therefore,
if $X$ is not supported by at most $\delta\Delta/2$ rows and at most
$\delta\Delta/2$ columns, it must fall into a subset $\X_P$ of matrices as
defined in Lemma~\ref{lem:XP}. 
So we apply Lemma~\ref{lem:XP} and count the expected number of sets $\X_P$ that
contain a matrix in the dual tensor code. The number of sets $\X_P$ is not more
than the number $\binom{\Delta}{p}\leq 2^{\Delta}$ of ways of choosing $P$,
times the number of ways of choosing $X_p$, which is not more than 
$\binom{\Delta}{\sqrt{\Delta}/\log_2\Delta}^{\Delta}\leq
2^{\Delta^{3/2}}.$ Therefore the number of $\X_P$'s is 
$2^{O(\Delta^{3/2})}$ while the probability that any given $\X_P$
contains a matrix in the dual tensor code is
$1/2^{\Omega(\Delta^{3/2}\log{\Delta})}$. Hence the result.
\end{proof}

Finally we shall need this last lemma.

\begin{lemma}
\label{lem:d/2}
Let $c$ be a codeword of weight $w$ of the dual tensor code $C_A\otimes\f_2^B+\f_2^A\otimes C_B$
where $C_A$ and $C_B$ have minimum distances $d_A$ and $d_B$. If $c$ is
supported by at most $d_A/2$ rows and at most $d_B/2$ columns, then it is
supported by at most $w/d_A$ columns and at most $w/d_B$ rows.
\end{lemma}

\begin{proof}
Write $c=\mathbf{r}+\mathbf{c}$ where $\mathbf{r}\in \f_2^A\otimes C_B$ is a vector supported by at most
$d_A/2$ rows, all of which are in $C_B$, and where likewise $\mathbf{c}\in
C_A\otimes\f_2^B$ has at most $d_B/2$ non-zero columns, all of which are in
$C_A$.
We have obviously that $\mathbf{r}$ is supported by at most $|\mathbf{r}|/d_B$
rows and that $\mathbf{c}$ is supported by at most $|\mathbf{c}|/d_A$ columns.
Since we clearly have $w=|\mathbf{r}+\mathbf{c}|\geq |\mathbf{r}|$ and likewise $w\geq
|\mathbf{c}|$, the result follows.
\end{proof}

We are now ready to prove Theorem~\ref{thm:P-robust}, which we recall.

\Probust*

\begin{proof}
Let us fix $\delta <\min(h^{-1}(1-\rho_A),h^{-1}(1-\rho_B))$. By standard union
bound arguments, we have that the minimum distances $d_A$ and $d_B$ of $C_A$ and
$C_B$ satisfy $\min(d_A,d_B)>\delta\Delta$ with probability
$1-1/2^{\Omega(\Delta)}$. Set $w=\Delta^{3/2-\eps}$: we have that $w<\frac\delta
2\Delta^{3/2}/\log_2\Delta$ for $\Delta$ large enough. Applying
Theorem~\ref{thm:robustness-dual}, we have that with probability
	$1-1/2^{\Omega(\Delta)}$, every codeword $c$ of weight $|c|\leq w$ of
the dual tensor code $C_A\otimes\f_2^B+\f_2^A\otimes C_B$
is supported by at most $d_A/2$ rows and at most $d_B/2$ columns, from which we
infer, applying Lemma~\ref{lem:d/2}, that $c$ is supported by at most $|c|/d_A$ columns and
at most $|c|/d_B$ rows. In other words the dual tensor code 
$C_A\otimes\f_2^B+\f_2^A\otimes C_B$ is $w$-robust with probability $1-1/2^{\Omega(\Delta)}$.

To obtain resistance to puncturing, we repeat the argument for every pair of
subsets $A'\subset A$, $B'\subset B$, $|A'|=|B'|\geq\Delta - \Delta^\gamma$. We
will obtain that every dual tensor code $C_{A'}\otimes\f_2^{B'}+\f_2^{A'}\otimes
C_{B'}$ is {\em not} $w$-robust with probability $1/2^{\Omega(\Delta)}$, from
which it will follow that all such dual tensor codes are $w$-robust, except with
probability
\[
\sum_{i\leq\Delta^\gamma}\binom{\Delta}{i}^2\frac{1}{2^{\Omega(\Delta)}}=\frac{1}{2^{\Omega(\Delta)}}.
\]
We note that, for $|A'|=|B'|=\Delta'$, the punctured code $C_{A'}$ is obtained
from a random uniform
generator matrix of size $\rho_A\Delta\times\Delta'$ and the punctured code
$C_{B'}$ is obtained by first choosing a random variable $r\geq
\Delta'-\rho_B\Delta$ and then, a random uniform $r\times\Delta'$ parity-check
matrix for $C_{B'}$. Since $\Delta'/\Delta$ must tend to $1$ when $\Delta$ tends
to infinity, this justifies applying Theorem~\ref{thm:robustness-dual} to every
$C_{A'}\otimes\f_2^{B'}+\f_2^{A'}\otimes C_{B'}$, as just argued.
\end{proof}

\newpage

\begin{thebibliography}{DEL{\etalchar{+}}21}

\bibitem[AAV13]{AAV13}
Dorit Aharonov, Itai Arad, and Thomas Vidick.
\newblock {Guest column: the quantum PCP conjecture}.
\newblock {\em ACM SIGACT news}, 44(2):47--79, 2013.

\bibitem[AE15]{AE15}
Dorit Aharonov and Lior Eldar.
\newblock Quantum locally testable codes.
\newblock {\em SIAM Journal on Computing}, 44(5):1230--1262, 2015.

\bibitem[BE21a]{BE21b}
Nikolas~P Breuckmann and Jens~N Eberhardt.
\newblock Balanced product quantum codes.
\newblock {\em IEEE Transactions on Information Theory}, 67(10):6653--6674,
  2021.

\bibitem[BE21b]{EB21}
Nikolas~P Breuckmann and Jens~N Eberhardt.
\newblock Quantum low-density parity-check codes.
\newblock {\em PRX Quantum}, 2:040101, Oct 2021.

\bibitem[BFLS91]{BFL91}
L\'{a}szl\'{o} Babai, Lance Fortnow, Leonid~A Levin, and Mario Szegedy.
\newblock Checking computations in polylogarithmic time.
\newblock In {\em Proceedings of the Twenty-Third Annual ACM Symposium on
  Theory of Computing}, STOC '91, page 21–32, New York, NY, USA, 1991.
  Association for Computing Machinery.

\bibitem[CS96]{CS96}
Robert Calderbank and Peter~W Shor.
\newblock Good quantum error-correcting codes exist.
\newblock {\em Physical Review A}, 54:1098--1105, Aug 1996.

\bibitem[DEL{\etalchar{+}}21]{DEL21}
Irit Dinur, Shai Evra, Ron Livne, Alexander Lubotzky, and Shahar Mozes.
\newblock {Locally Testable Codes with constant rate, distance, and locality}.
\newblock {\em arXiv preprint arXiv:2111.04808}, 2021.

\bibitem[EH17]{EH17}
Lior Eldar and Aram~W Harrow.
\newblock {Local Hamiltonians whose ground states are hard to approximate}.
\newblock In {\em Foundations of Computer Science (FOCS), 2017 IEEE 58th Annual
  Symposium on}, pages 427--438. IEEE, 2017.

\bibitem[EKZ20]{EKZ20}
Shai {Evra}, Tali {Kaufman}, and Gilles {Zémor}.
\newblock {Decodable quantum LDPC codes beyond the square root distance barrier
  using high dimensional expanders}.
\newblock In {\em 2020 IEEE 61st Annual Symposium on Foundations of Computer
  Science (FOCS)}, pages 218--227, 2020.

\bibitem[FML02]{FML02}
Michael~H Freedman, David~A Meyer, and Feng Luo.
\newblock {$Z_2$-systolic freedom and quantum codes}.
\newblock {\em Mathematics of quantum computation, Chapman \& Hall/CRC}, pages
  287--320, 2002.

\bibitem[Gal62]{Gal62}
Robert Gallager.
\newblock Low-density parity-check codes.
\newblock {\em IRE Transactions on Information Theory}, 8(1):21--28, 1962.

\bibitem[Gol10]{Gol10}
Oded Goldreich.
\newblock {\em Short Locally Testable Codes and Proofs: A Survey in Two Parts},
  pages 65--104.
\newblock Springer Berlin Heidelberg, Berlin, Heidelberg, 2010.

\bibitem[Got97]{got97}
Daniel Gottesman.
\newblock {\em Stabilizer codes and quantum error correction}.
\newblock PhD thesis, California Institute of Technology, 1997.

\bibitem[GS06]{GS06}
Oded Goldreich and Madhu Sudan.
\newblock {Locally Testable Codes and PCPs of Almost-Linear Length}.
\newblock {\em J. ACM}, 53(4):558–655, jul 2006.

\bibitem[Gur10]{Gur}
Venkatesan Guruswami.
\newblock Expander codes and their decoding.
\newblock
  \url{https://www.cs.cmu.edu/~venkatg/teaching/codingtheory/notes/notes8.pdf},
  2010.

\bibitem[Has13]{has13}
Matthew~B Hastings.
\newblock {Trivial low energy states for commuting Hamiltonians, and the
  quantum PCP conjecture}.
\newblock {\em Quantum Information \& Computation}, 13(5-6):393--429, 2013.

\bibitem[Has17]{has17}
Matthew~B Hastings.
\newblock Quantum codes from high-dimensional manifolds.
\newblock In {\em 8th Innovations in Theoretical Computer Science Conference
  (ITCS 2017)}. Schloss Dagstuhl-Leibniz-Zentrum fuer Informatik, 2017.

\bibitem[Has21]{has21}
Matthew~B Hastings.
\newblock On quantum weight reduction.
\newblock {\em arXiv preprint arXiv:2102.10030}, 2021.

\bibitem[HHO20]{HHO20}
Matthew~B Hastings, Jeongwan Haah, and Ryan O’Donnell.
\newblock {Fiber Bundle Codes: Breaking the $N^{1/2} \mathrm{polylog}(N)$
  Barrier for Quantum LDPC Codes}.
\newblock {\em arXiv preprint arXiv:2009.03921}, 2020.

\bibitem[HLW06]{HLW06}
Shlomo Hoory, Nathan Linial, and Avi Wigderson.
\newblock Expander graphs and their applications.
\newblock {\em Bulletin of the American Mathematical Society}, 43(4):439--561,
  2006.

\bibitem[Kit03]{kit03}
Alexei~Yu Kitaev.
\newblock Fault-tolerant quantum computation by anyons.
\newblock {\em Annals of Physics}, 303(1):2--30, 2003.

\bibitem[KT20]{KT20}
Tali Kaufman and Ran~J Tessler.
\newblock {Quantum LDPC codes with $\Omega(\sqrt{n}\log^k n)$ distance, for any
  $k$}.
\newblock {\em arXiv preprint arXiv:2008.09495}, 2020.

\bibitem[LH22]{LH22}
Ting-Chun Lin and Min-Hsiu Hsieh.
\newblock $c^3$-local testable codes from lossless expanders.
\newblock {\em arXiv preprint arXiv:2201.11369}, 2022.

\bibitem[LLZ21]{LLZ21}
Anthony Leverrier, Vivien Londe, and Gilles Z\'{e}mor.
\newblock {Towards Local Testability for Quantum Coding}.
\newblock In James~R. Lee, editor, {\em 12th Innovations in Theoretical
  Computer Science Conference (ITCS 2021)}, volume 185, pages 65:1--65:11.
  Schloss Dagstuhl--Leibniz-Zentrum f{\"u}r Informatik, 2021.

\bibitem[PK20]{PK20}
Pavel Panteleev and Gleb Kalachev.
\newblock {Quantum LDPC Codes with Almost Linear Minimum Distance}.
\newblock {\em arXiv preprint arXiv:2012.04068}, 2020.

\bibitem[PK21]{PK21}
Pavel Panteleev and Gleb Kalachev.
\newblock {Asymptotically good quantum and locally testable classical LDPC
  codes}.
\newblock {\em arXiv preprint arXiv:2111.03654}, 2021.

\bibitem[PS94]{PS94}
Alexander Polishchuk and Daniel~A Spielman.
\newblock Nearly-linear size holographic proofs.
\newblock In {\em Proceedings of the twenty-sixth annual ACM symposium on
  Theory of computing}, pages 194--203, 1994.

\bibitem[SS96]{SS96}
Michael Sipser and Daniel~A. Spielman.
\newblock Expander codes.
\newblock {\em IEEE Transactions on Information Theory}, 42(6):1710--1722,
  1996.

\bibitem[Ste96]{ste96}
Andrew~M. Steane.
\newblock Error correcting codes in quantum theory.
\newblock {\em Physical Review Letters}, 77:793--797, Jul 1996.

\bibitem[Tan81]{T81}
R.~Tanner.
\newblock A recursive approach to low complexity codes.
\newblock {\em IEEE Transactions on Information Theory}, 27(5):533--547, 1981.

\bibitem[TZ14]{TZ13}
Jean-Pierre Tillich and Gilles Zémor.
\newblock {Quantum LDPC codes with positive rate and minimum distance
  proportional to the square root of the blocklength}.
\newblock {\em IEEE Transactions on Information Theory}, 60(2):1193--1202,
  2014.

\end{thebibliography}
%

\newcommand{\etalchar}[1]{$^{#1}$}

\end{document}